\documentclass[twocolumn,amssymb, nobibnotes, aps, prd, superscriptaddress]{revtex4-1}

\newcommand{\vect}[1]{\boldsymbol{#1}}
\usepackage{graphicx}
\usepackage{amsmath, amsthm, amssymb}
\theoremstyle{definition}

\usepackage{mathdots}
\usepackage{float}
\usepackage{pgfplots} 
\usepackage{tikz-3dplot}
\usepgfplotslibrary{patchplots}
\usepgfplotslibrary{groupplots}
\pgfplotsset{compat=1.6,ylabsh/.style={every axis y label/.style={at={(0,0.5)}, xshift=#1, rotate=90}}}
\usetikzlibrary{positioning}
\usetikzlibrary{shapes}
\usetikzlibrary{shapes.geometric}
\usetikzlibrary{decorations.text}
\usetikzlibrary{decorations.pathreplacing,shapes.misc}
\usetikzlibrary{matrix}
\usepackage{caption,subcaption}
\usepackage[ruled,lined]{algorithm2e}
\usepackage{algorithmic}
\renewcommand{\algorithmiccomment}[1]{\bgroup\hfill//~#1\egroup}

\providecommand{\norm}[1]{\lVert#1\rVert}
\providecommand{\bracket}[1]{\left(#1\right)}
\newenvironment{frcseries}{\fontfamily{frc}\selectfont}{}

\begin{document}
	
\title{Cell strain-stiffening drives cell breakout from embedded spheroids}
\author{Shabeeb Ameen}
\email{mameen@syr.edu}
\affiliation{Physics Department, Syracuse University, Syracuse, NY 13244 USA}
\author{Kyungeun Kim}
\email{kkim@math.ubc.edu}
\affiliation{Department of Mathematics, University of British Columbia, Vancouver, BC V6T 1Z2, Canada}
\author{Ligesh Theeyancheri}
\email{ltheeyan@syr.edu}
\affiliation{Physics Department, Syracuse University, Syracuse, NY 13244, USA}
\author{Minh Thanh}
\email{mohthanh@syr.edu}
\affiliation{Physics Department, Syracuse University, Syracuse, NY 13244, USA}
\author{Mingming Wu}
\email{mw272@cornell.edu}
\affiliation{Department of Biological and Environmental Engineering, Cornell University, Ithaca, NY 14853, USA}
\author{Alison E. Patteson}
\email{aepattes@syr.edu}
\affiliation{Physics Department, Syracuse University, Syracuse, NY 13244, USA}
\author{J. M. Schwarz}
\email{jmschw02@syr.edu}
\affiliation{Physics Department, Syracuse University, Syracuse, NY 13244, USA}
\affiliation{Indian Creek Farm, Ithaca, NY 14850, USA}
\author{Tao Zhang}
\email{zhangtao.scholar@sjtu.edu.cn}
\affiliation{School of Chemistry and Chemical Engineering, Shanghai Jiao Tong University, Shanghai 200240, China}

\date{\today}
\begin{abstract}

Understanding how cells escape from embedded spheroids requires a mechanical framework that connects stress generation within cells, across cells, and between cells and the surrounding extracellular matrix (ECM). Here, we develop such a framework using a 3D vertex model of a spheroid coupled to a fibrous ECM network and derive the 3D Cauchy stress tensor for deformable polyhedral cells, enabling direct quantification of cell-level stresses in three dimensions beyond shape-based or bulk rheological measures. We characterize maximum shear stress distributions for solid-like and fluid-like spheroids. Solid-like spheroids exhibit broad stress distributions and spatial stress gradients, whereas fluid-like spheroids display lower stresses and minimal spatial patterning. We find little generic alignment between cell shape anisotropy and principal stress directions, demonstrating that morphology alone is not a reliable indicator of mechanical state. We further show that individual cells undergo strain stiffening upon elongation, producing nonlinear increases in maximum shear stress. This provides a mechanism by which boundary cells in otherwise low-stress, fluid-like spheroids can transiently generate forces sufficient to remodel the surrounding matrix. These results indicate that invasion cannot be understood solely through tissue-level unjamming transitions, but also involves cell-level stress amplification processes. To investigate how this strain-induced stress amplification couples to collective invasion modes, we introduce an extended 3D vertex model in which cells interact through explicit, tunable cell–cell adhesion springs. Within this minimal mechanical framework, single-cell breakout arises from the combination of strain stiffening and weakened cell–cell adhesion, whereas multi-cell streaming requires an additional ingredient: anisotropic adhesion strengthened along the axis of elongation in highly strained cells and weakened orthogonally. These findings identify distinct mechanical pathways linking cell strain, stress amplification, and adhesion organization to different modes of spheroid invasion.
\end{abstract}		
\maketitle
		
\section{Introduction}
Predicting the invasive potential of a solid tumor requires understanding the interplay between cell-cell interactions and cell-extracellular matrix (ECM) interactions. To quantify this interplay, researchers use {\it in vitro} models of solid tumors, known as embedded spheroids consisting of a cellular collective embedded within a collagen matrix~\cite{friedl1998}. Specifically, experimental studies have demonstrated that factors such as collagen fiber density can dramatically impact the rate of spheroid cells invading the surrounding collagen with cells much less likely to invade at higher collagen densities~~\cite{Huang2020}.  Other experiments illustrate the importance of the cytoskeletal filament vimentin in helping drive cell invasion via a combination of spheroid fluidization and degradation of the collagen~\cite{Ho2024}. Additional experiments show with increasing uni-axial compression that the cells are more likely to invade the surrounding collagen~\cite{Pandey2024,Pandey2025}. 

In addition to determining whether or not the cells invade the surrounding environment, experiments exhibit different types of cell invasion modes, such as single-cell invasion, a stream of cells invading, and a collective front of cells invading its surroundings. Prior theoretical work argues that the different invasion modes are driven by confinement and cell-cell adhesion~\cite{Ilina2020}. For example, high cell-cell adhesion lends to more coordinated cell breakout with low confinement as compared to low cell-cell adhesion leading to individual cell breakout. These findings are based on cellular automaton simulations that do not contain cell deformations or stresses. More recent computational work implements a 2D Voronoi model to demonstrate that higher compressive stress exerted on the spheroid by the environment prevents invasion and low cell-cell adhesion promotes invasion~\cite{Lui_2017}. Finally, the interplay between cell phenotype switching and ECM remodeling to determine invasion profiles has been approached as a coupled dynamical system in which cells move through a 3D ECM while switching between multiple motility phenotypes, with each phenotype having its own migration speed and matrix-remodeling capability, though there is no explicit mechanics in the model~\cite{kim2007}. Given the ever-increasing number of experimental and computational results, there is a compelling need for a 3D computational framework involving both the spheroid and the fibrous network and their respective mechanics to search for guiding principles for experiments. 

\begin{figure*}[t]
             \includegraphics[width=0.49\textwidth]{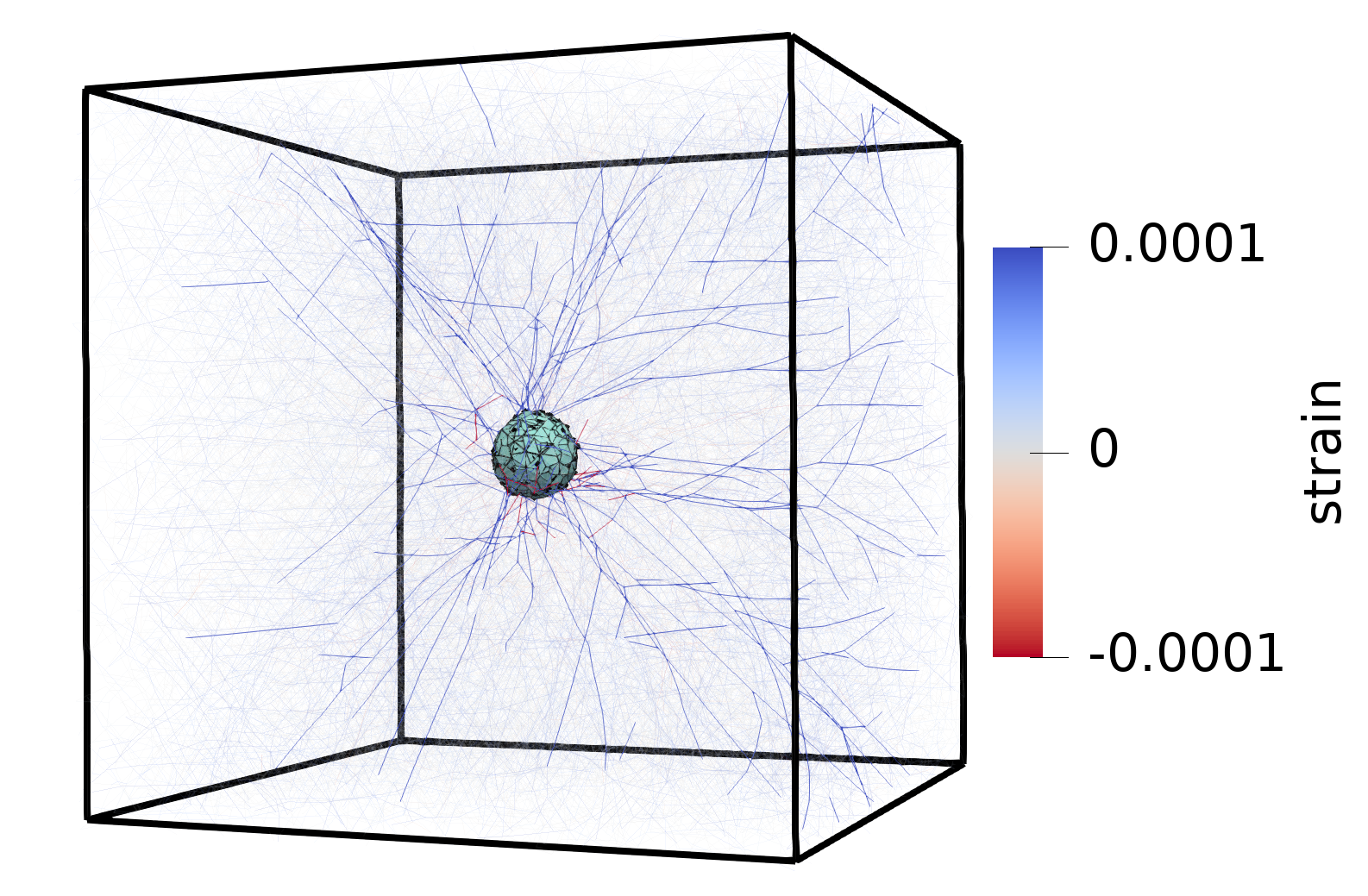}  
            \includegraphics[width=0.42\textwidth]{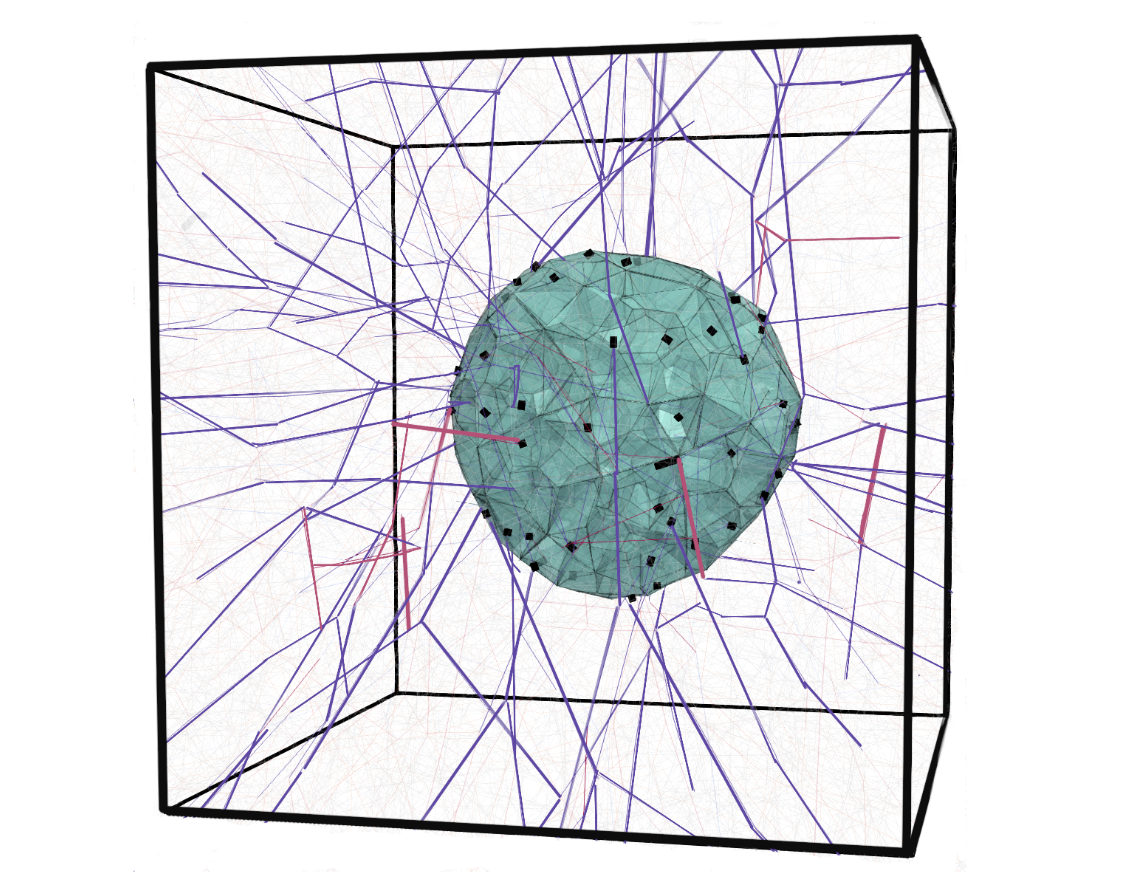}
             \label{fig: simulation_snapshots}
         \caption{{\it Integrating a 3D vertex model of a cell spheroid with a fiber-network model to study cell--ECM interactions. }  Left: Full system at the final simulation time, i.e., time $t = t_f$ and Right: Zoom in of the spheroid. The black denote active linker springs coupling the cells to the fibers. }
\end{figure*}
A recent computational model that couples a three-dimensional vertex model to a three-dimensional fiber network via active linker springs provides such a starting point (see Fig. 1)~\cite{Zhang2022,zhang2025}. With this starting point, both the rheology of the spheroid and the fiber network stiffness affect the remodeling of the fiber network. Specifically, fluid-like spheroids densify and radially realign the fiber network more on average than solid-like spheroids for a range of intermediate fiber network stiffnesses. Their predictions were supported by experimental studies comparing non-tumorigenic MCF10A spheroids and malignant MDA-MB-231 spheroids embedded in collagen networks. The spheroid rheology-dependent effects are the result of cellular motility generating spheroid shape fluctuations. These shape fluctuations lead to emergent feedback between the spheroid and the fiber network to further remodel the fiber network. This emergent feedback occurs only at intermediate fiber network stiffness. At low fiber network stiffness, the mechanical response of the coupled system is dominated by the spheroid, while at high fiber network stiffness, the mechanical response is dominated by the fiber network. They, therefore, quantified the notion of optimal spheroid-fiber network mechanical reciprocity. 

Using this computational model as a foundation, we now explore the role of cell stresses, in addition to cell shape. We anticipate an important role for cell stress particularly in cells invading the surrounding fiber network, which we will call from now on cell breakout. To calculate cell stress in three dimensions, we extend for the first time an earlier two-dimensional cell stress tensor formulation in bulk tissue \cite{nestor2018}. With this cell stress tensor we can determine how the cell stresses are distributed in the embedded spheroid and ask: How does this distribution change with the fluidity of the embedded spheroid? Moreover, since fluid-like spheroids remodel the fiber network more strongly than solid-like spheroids, one would expect fluid-like cells to be more likely to break out of the spheroid, and yet fluid-like cells are very low stress cells. Such cells would not be able to remodel the fibers as they moved along them, should they break away from the spheroid. And what mechanisms distinguish single cell breakout from multiple cell breakout in the form of cell streaming, i.e. one cell following another? 

To begin to answer these questions, we first briefly revisit the computational model, then present the cell stress tensor and use it to compute the maximum shear stress, a scalar quantity. We will also compute another quantity, the cell shape anisotropy based on the conventional radius of gyration tensor. We will look for correlations in these quantities within the spheroid as well as any spatial patterning of the quantities within the spheroid. Finally, we will use this information to answer the seeming paradox of low stress fluid cells breaking out of the spheroid (since fluid-like spheroids can remodel the fiber network more than solid-like spheroids) and a new mechanism for single cell breakout versus cell streaming breakout. We will do so by introducing an extended vertex model that treats cell–cell adhesion in an explicit but coarse-grained manner. We Details of the three-dimensional computational model, including the new extended vertex model version, and the definitions of the cell stress tensor and cell shape measures are provided in the Supplemental Material.

\begin{figure*}[t]
         \centering
         \begin{subfigure}[b]{0.31\textwidth}
             \centering
             \includegraphics[width=\textwidth]{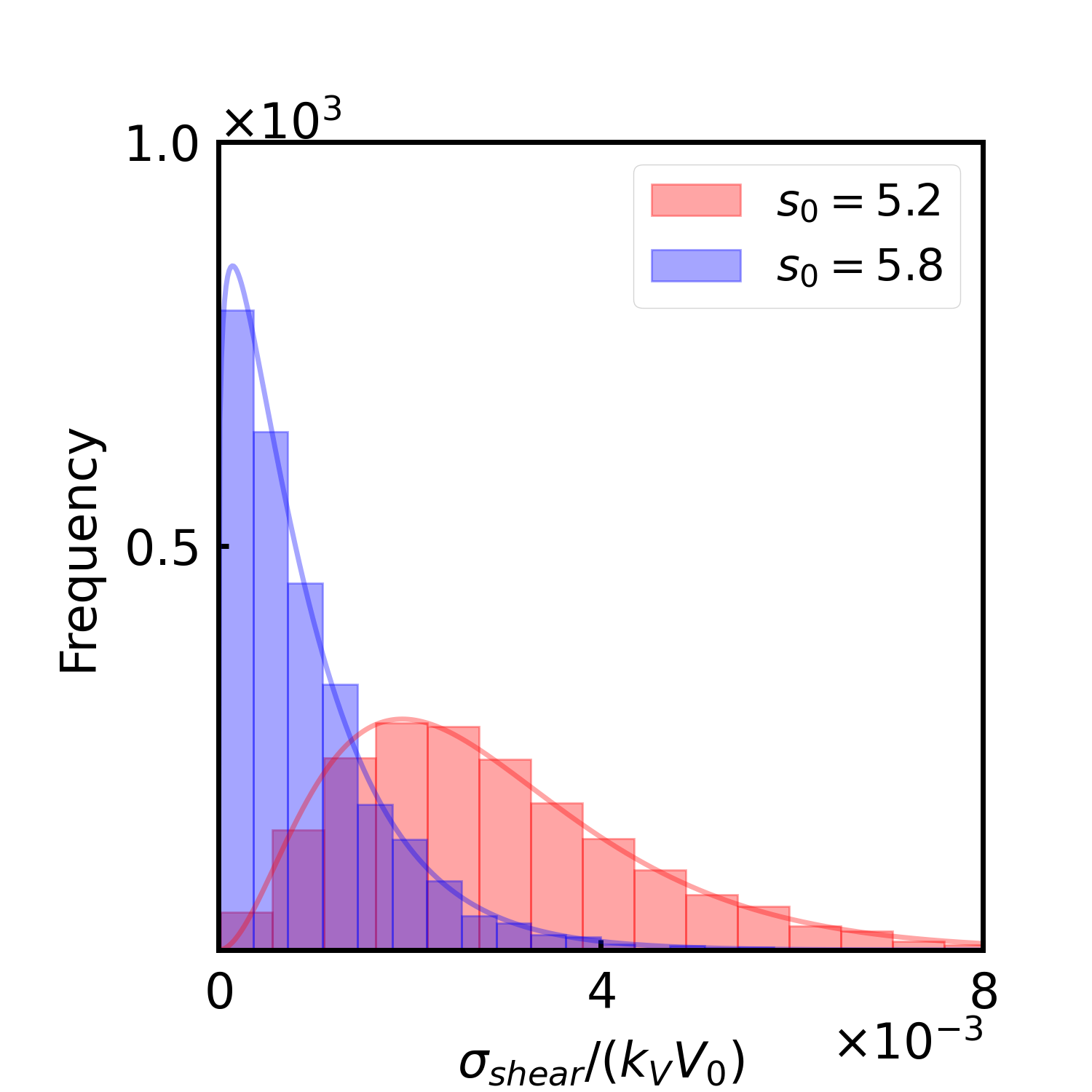}
             \caption{}
             \label{fig: simulation_close_up}
         \end{subfigure}
         \begin{subfigure}[b]{0.31\textwidth}
             \centering
             \includegraphics[width=\textwidth]{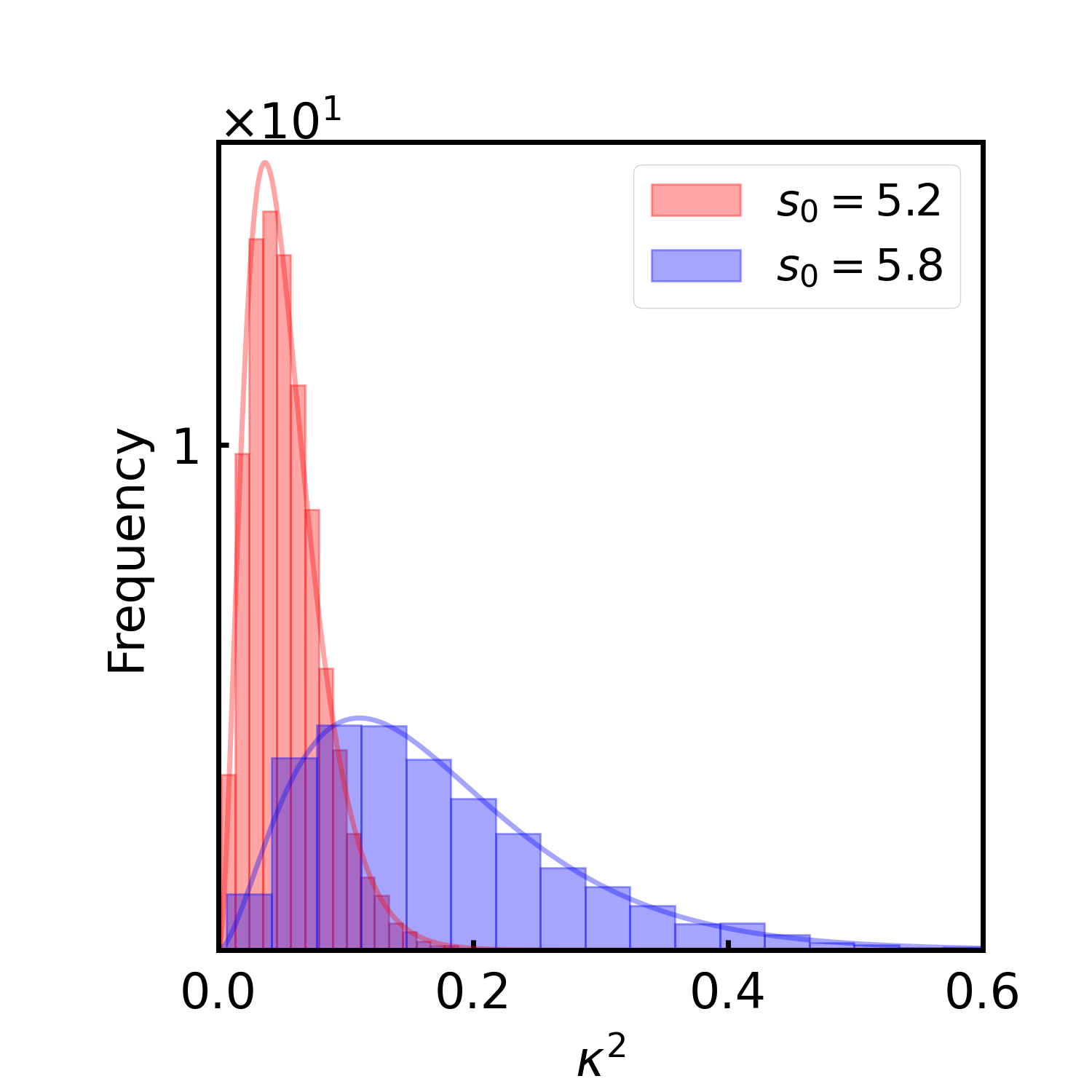}
             \caption{}
             \label{fig: cell_shape_anisotropy_histogram}
         \end{subfigure}        
         \begin{subfigure}[b]{0.31\textwidth}
             \centering
             \includegraphics[width=\textwidth]{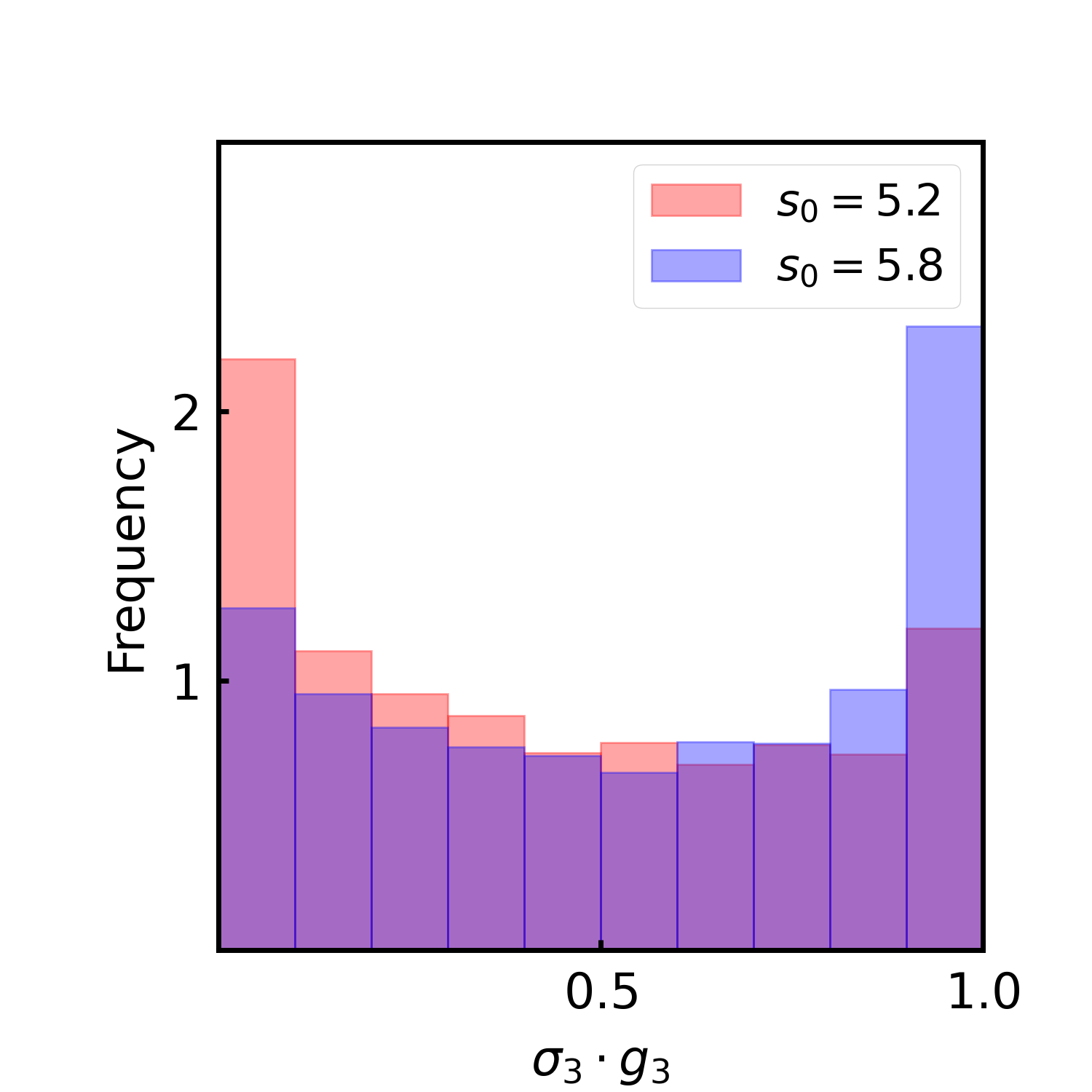}
             \caption{}
             \label{fig: stress_shape_overlap}
         \end{subfigure}
         \centering
         \caption{{\it Distributions of cellular shear stress, cell shape, and stress--shape correlation for solid-like and fluid-like spheroids at fiber-network occupation probability $p=0.8$.}
 (a) Histogram of the  maximum shear stress for the cells for two different target $s_0$s. Gamma fit parameters: $s_0 = 5.2$: $\alpha = 3.06, \theta = 9.34\times10^{-4}$; $s_0 = 5.8$: $\alpha = 1.19, \theta = 7.70\times10^{-4}$. (b) Histogram of the cell shape anisotropy for the cells for two different target $s_0$s. Gamma fit parameters: $s_0 = 5.2$: $\alpha = 3.20, \theta = 1.66\times10^{-2}$; $s_0 = 5.8$: $\alpha = 2.78, \theta = 6.21\times10^{-2}$.  (c) The overlap, or dot product, between the eigenvector associated with the largest maximum shear stress and the eigenvector associated with the largest gyration eigenvalue.  }
\end{figure*}

\section{Results}

Our results are based on a 3D Vertex Model representing the spheroid embedded in a disordered fibrous network mimicking the collagen matrix. Cells in the spheroid are coupled to the fiber network by active linker springs that contract over time. The construction of this coupled spheroid–matrix system, including energy functionals, the numerical integration procedure, and parameter choices, is described in detail in the 3D Computational Model subsection of the Supplemental Material.

\subsection{Maximum shear stress $\sigma_{\rm shear}$ and its correlation with cell shape anisotropy $\kappa^2$}

For each polyhedral cell, we compute the Cauchy stress tensor $\boldsymbol{\sigma}$ and define the \emph{maximum shear stress}
$\sigma_{\rm shear} = (\sigma_3-\sigma_1)/2$, where $\sigma_1 \le \sigma_2 \le \sigma_3$ are the principal stresses (see Eq.~(\ref{maximum_shear_stress}) in the Supplemental Material). Cell shape is quantified by the gyration tensor $\boldsymbol{G}$ and the corresponding dimensionless \emph{cell shape anisotropy} $\kappa^2$, which ranges from 0 for nearly spherical cells to 1 for highly elongated cells (see Eq.~(\ref{shape_tensor}) and the discussion in the ``Cell Properties: Stress and Shape'' subsection of the Supplemental Material).

Figure~2(a) shows the distribution of $\sigma_{\rm shear}$ for solid-like spheroids ($s_0=5.2$) and fluid-like spheroids ($s_0=5.8$) embedded in a matrix with bond occupation probability $p=0.8$. Solid-like spheroids display a broader distribution with a larger mean $\sigma_{\rm shear}$ than fluid-like spheroids, consistent with higher internal stresses in the solid-like state. By construction, $\sigma_{\rm shear}\ge 0$. Moreover, earlier work showed that the energy-barrier-to-cell-rearrangement distributions in a two-dimensional version of this vertex model were well-characterized by a Gamma function distribution,~\cite{Bi_2014}. Therefore, to better quantify the differences in the maximum shear stress distribution for different system parameters, namely $s_0$, the target cell shape index and $p$, the fiber network occupation probability, which also encodes network stiffness, we fit the histograms to Gamma function distributions. The fit parameters are reported in the caption of Fig.~2.

Within this framework, the shape parameter $\alpha$ provides a measure of the degree of mechanical heterogeneity in the cell population, while the scale parameter $\theta$ sets the characteristic stress scale of individual contributions. Solid-like spheroids exhibit larger $\alpha$ values and broader distributions, indicating a wide spectrum of stress states and pronounced mechanical heterogeneity. This is consistent with the presence of spatial stress gradients and long-lived force chains in the solid-like regime. In contrast, fluid-like spheroids display smaller $\alpha$ values and narrower distributions, reflecting more homogeneous stress states arising from frequent cell rearrangements and efficient stress relaxation.

As the surrounding matrix is softened (e.g., $p=0.75$), the $\sigma_{\rm shear}$ distribution for solid-like spheroids does not change much from $p=0.8$. However, for isolated fluid-like spheroids without a fibrous network, there is additional maximum shear stress broadening. The corresponding histograms are shown in Figs.~S1 and S2 of the Supplemental Material. For intermediate ranges of matrix stiffness, the active pulling on the compliant matrix provides an additional pathway for stress relaxation. When the matrix is stiffened by increasing $p$ (or there is the matrix is floppy or absent, this relaxation pathway is suppressed and the $\sigma_{\rm shear}$ distribution broadens. Thus, the regime of optimal spheroid–matrix mechanical reciprocity coincides with the narrowest maximum shear stress distribution for the fluid-like spheroids. In contrast, fluid-like spheroids exhibit comparatively small $\sigma_{\rm shear}$ values over the full range of $p$s.

Looking to cell breakout, the Gamma-distributed nature of the maximum shear stress highlights that invasion-relevant stresses arise from the tails of the distribution rather than its mean. While fluid-like spheroids are characterized by lower average stresses, the exponential tail of the Gamma distribution allows for rare but mechanically significant high-stress cells. These cells become especially relevant at the spheroid boundary, where elongation-induced strain stiffening can dramatically amplify maximum shear stress and enabling force transmission into the surrounding fiber network. Thus, the Gamma fits quantitatively capture how rare stress fluctuations, rather than bulk stress levels, presumably govern the onset of cell breakout.

Let us now further characterize the cells by their shape anisotropy. The distributions of $\kappa^2$ for the same cells are plotted in Fig.~2(b). Fluid-like spheroids ($s_0=5.8$) exhibit broader $\kappa^2$ distributions with larger mean anisotropy than solid-like spheroids, consistent with their larger cell shape index. To characterize these distributions, we find that the Gamma distribution also approximates well the shape anisotropy distribution. 

To quantify the overlap between stress and shape, we compute for each cell the dot product between the eigenvector associated with the largest principal stress $\sigma_3$ and the eigenvector associated with the largest gyration eigenvalue $g_3$, i.e., the long axis of the cell. The resulting overlap distribution is shown in Fig.~2(c). For solid-like spheroids, the distribution peaks near zero, indicating little correlation between the directions of maximal stress and maximal elongation; highly stressed cells are not necessarily less globular. Indeed, a stress tensor does not necessarily give information about shape and vice versa. By contrast, fluid-like spheroids display a distribution biased toward larger overlaps, demonstrating that the relatively rare high-stress cells in the fluid-like regime tend to be more elongated and that their principal stress aligns with their long axis.

\begin{figure*}[t]
         \centering
         \begin{subfigure}[b]{0.33\textwidth}
             \centering
             \includegraphics[width=\textwidth]{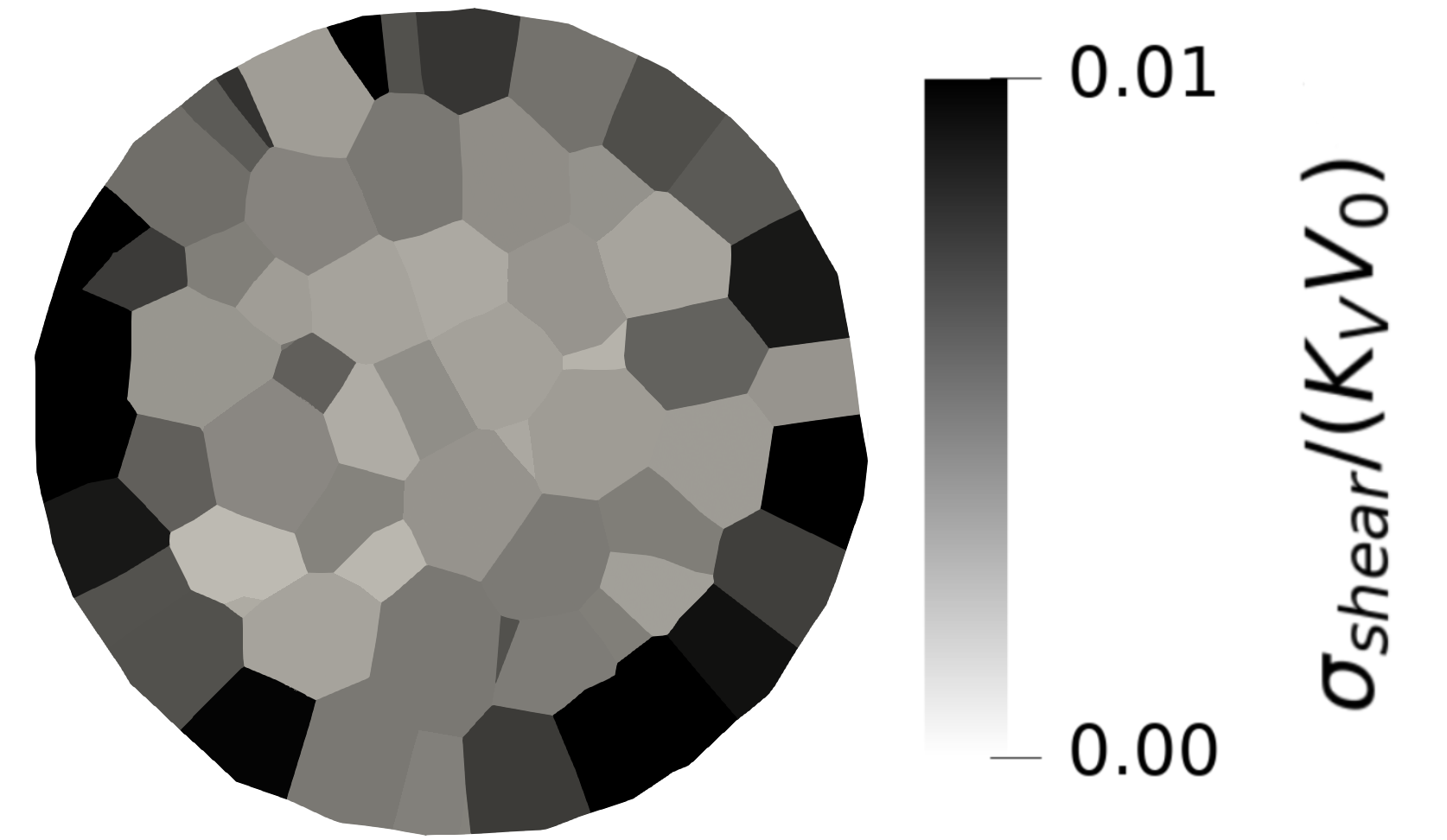}
             \caption{}
             \label{fig: shear_cross_section}
         \end{subfigure}
         \begin{subfigure}[b]{0.33\textwidth}
             \centering
             \includegraphics[width=\textwidth]{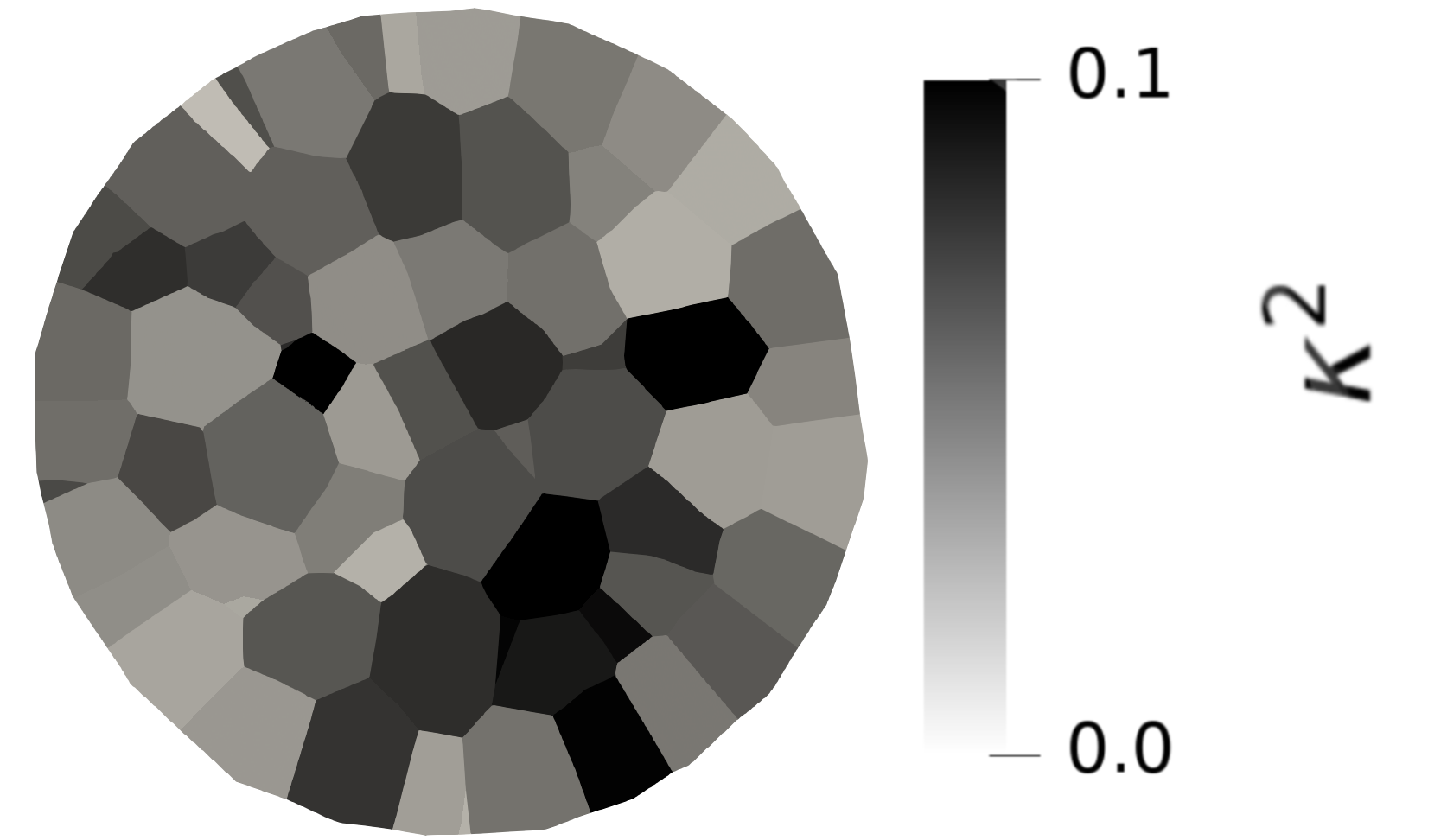}
             \caption{}
             \label{fig: anisotropy_cross_section}
         \end{subfigure}        
         \begin{subfigure}[b]{0.3\textwidth}
             \centering
             \includegraphics[width=\textwidth]{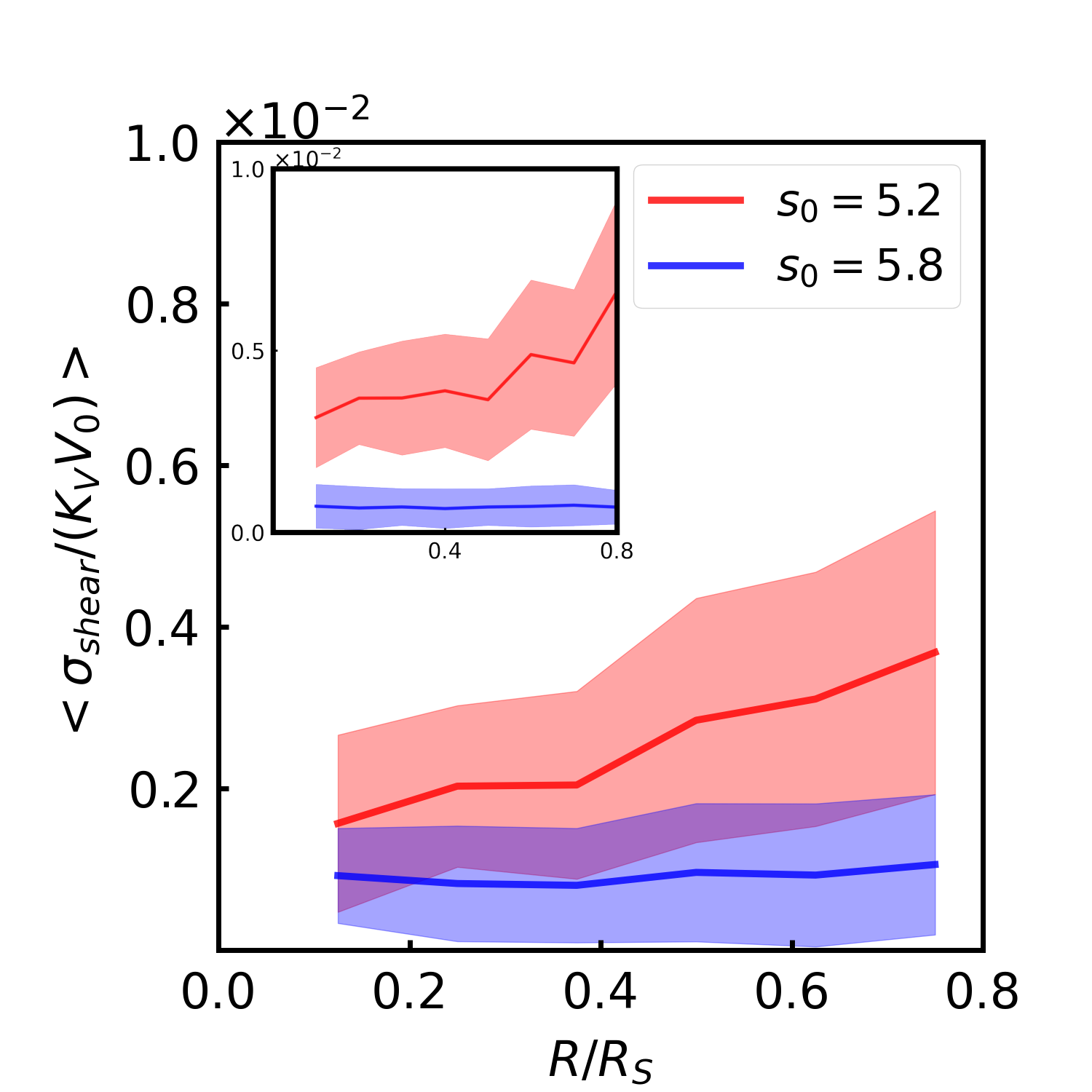}
             \caption{}
             \label{fig: radial_distribution_stress}
         \end{subfigure}
         \caption{{\it Spatial distribution of cellular stresses and cellular shape in solid-like spheroids.} (a) Sample cross-section of the cellular maximum shear stress. (b) Sample cross-section of the cellular shape anisotropy. (c) Maximum shear stress for both fluid-like and solid-like spheroids as a function of radial distance from the center of mass of the spheroid. The inset shows the result for  $p=0$, i.e., the non-embedded spheroid with no surrounding fiber network. }
\end{figure*}

\subsection{Spatial patterning of cell stress in solid-like spheroids}

We next examine how cell stresses and shapes are organized within the
spheroid. For solid-like spheroids ($s_0=5.2$), a central cross-section
shows that the maximum shear stress $\sigma_{\rm shear}$ is lowest
near the center and increases toward the boundary [Fig.~3(a)]. The
corresponding cross-section of the cell shape anisotropy $\kappa^2$
reveals that interior cells tend to be more anisotropic than boundary
cells [Fig.~3(b)]. Radial averaging confirms a monotonic increase of the
mean $\langle\sigma_{\rm shear}\rangle$ with distance from the spheroid
center [Fig.~3(c)]. This stress differential could lead to interesting,
nontrivial effects within the bulk of the spheroid, with the inner cells
being closer to fluidization than the cells near the boundary.

In contrast, for fluid-like spheroids ($s_0=5.8$) the radial profile
$\langle\sigma_{\rm shear}\rangle(R)$ is much flatter [Fig.~3(c)],
indicating the absence of a pronounced bulk–boundary gradient. Cross-
sectional maps of both $\sigma_{\rm shear}$ and $\kappa^2$ likewise show
no obvious spatial patterning.

The emergence of a radial stress gradient in solid-like spheroids can be understood as a consequence of constrained stress redistribution in a mechanically rigid collective. In the solid-like regime, limited cell rearrangements inhibit stress relaxation in the bulk, causing stresses generated by cell–ECM coupling and active contractility to accumulate preferentially near the boundary. Interior cells, shielded from direct ECM interactions and unable to reorganize efficiently, remain comparatively low stress, while boundary cells bear the mechanical load required to balance spheroid–matrix forces. This mechanism is absent in fluid-like spheroids, where frequent neighbor exchanges allow stresses to be redistributed and homogenized throughout the tissue.

\begin{figure*}[t]\label{strained_cell_comparisons}
         \centering
         \begin{subfigure}[b]{0.2\textwidth}
             \centering
             \includegraphics[width=\textwidth]{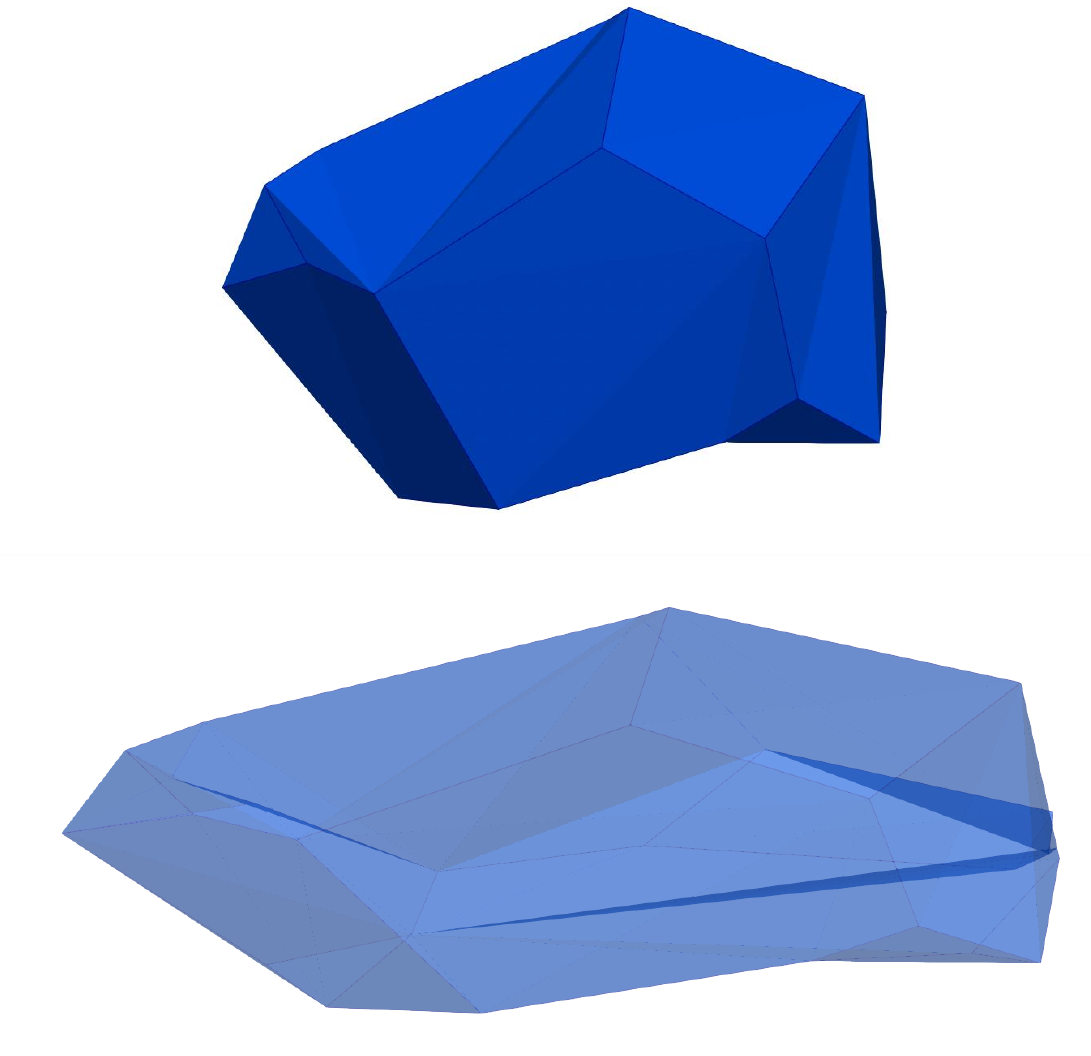}
             \caption{}
             \label{fig: final_cell}
         \end{subfigure}  
         \begin{subfigure}[b]{0.26\textwidth}
             \centering
             \includegraphics[width=\textwidth]{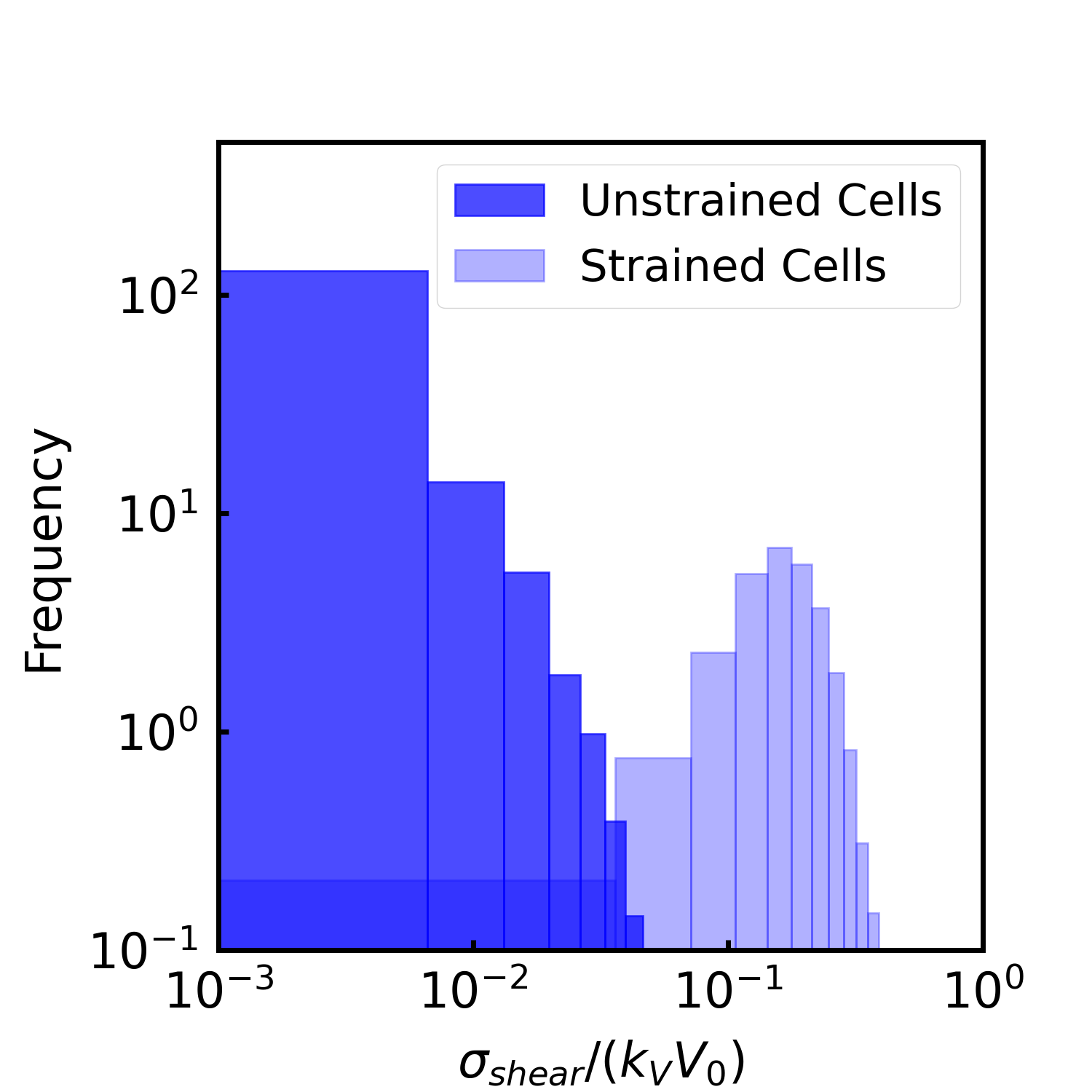}
             \caption{}
             \label{fig: single_cell_overlap_5.2}
         \end{subfigure}
         \begin{subfigure}[b]{0.26\textwidth}
             \centering
             \includegraphics[width=\textwidth]{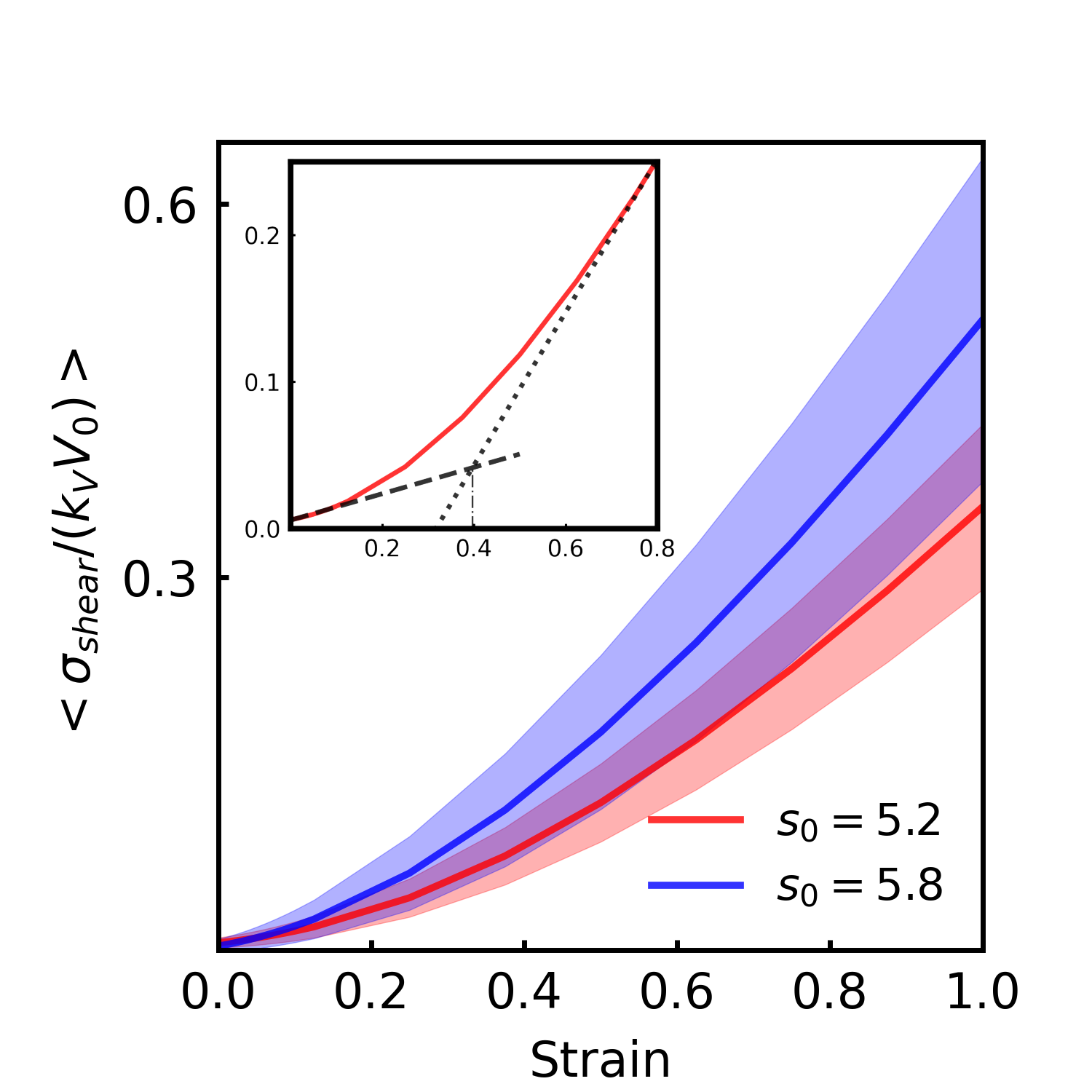}
              \caption{}
         \end{subfigure}
         \begin{subfigure}[b]{0.26\textwidth}
             \centering
             \includegraphics[width=\textwidth]{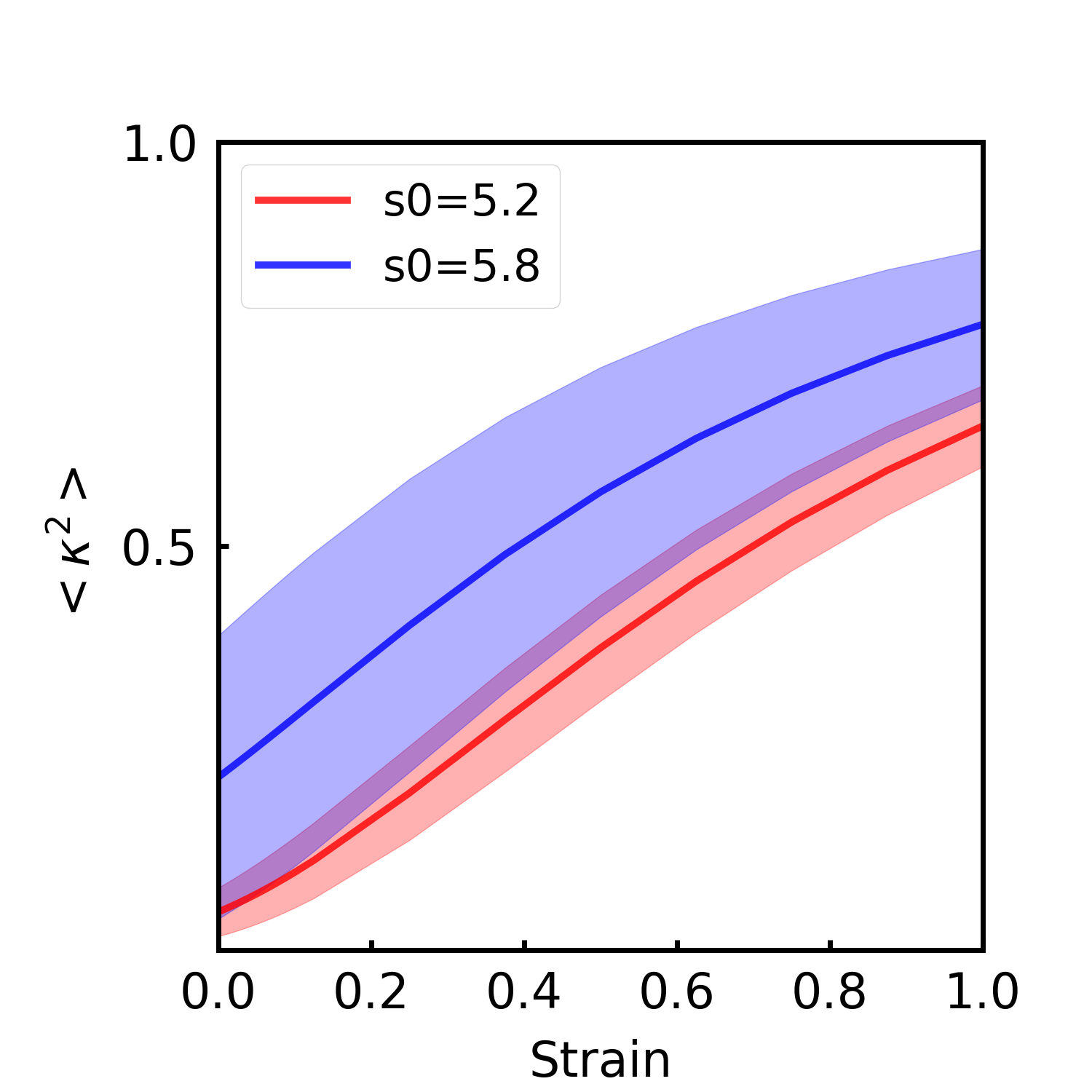}
              \caption{}
         \end{subfigure}
         \caption{{\it Cellular strain stiffening for volume-preserving deformations}. (a) Initial cell configuration and a strained cell configuration for cells from a fluid-like spheroid. (b) Distribution of maximum bulk shear stress for the initial cells (dark blue) and for the strained cells (light blue) for $s_0=5.8$. (c) The average maximum shear stress versus strain demonstrating a nonlinear relationship for both types of the spheroids and, thus, exhibiting strain stiffening phenomena. The inset shows a deviation from linear behavior around 0.1 strain, yielding an estimate of the onset of strain stiffening. By linearly fitting the low and high strain parts of the stress-strain curve, the intersection of the two fits yields a second estimate for the crossover between low and high strain behavior at approximately 0.4 (dashed vertical line).   (d) The average cell shape anisotropy versus strain.}
\end{figure*}

\subsection{Cell strain-stiffening with volume-preserving deformations}

Prior work demonstrates that fluid-like spheroids remodel the fiber network more so than solid-like spheroids~\cite{zhang2025}. Presumably, the enhanced remodeling leads to greater cell breakout potential for the fluid-like spheroids.  If so, the maximum shear stresses are much lower in the fluid-like case.  As stresses are needed to remodel the ECM as cells move along the fibers, how then can fluid-like cells generate enough stress to be able to do so, at least for mesenchymal-type motility?  To explain this conundrum of fluid-like cells moving out into the 3D fiber network, we investigate what happens when cells stretch themselves out, which is more likely for cells at the boundary of the spheroid as opposed to cells in the bulk. 
Specifically, we impose the following deformation: $\mathbf{r}_i' = U \, \lambda \, U^{T}\, \mathbf{r}_i$, where $U^{T} G = \operatorname{diag}(g_1, g_2, g_3)$ and $\lambda = \operatorname{diag}\!\left(\frac{1}{\sqrt{\alpha}},\,\frac{1}{\sqrt{\alpha}},\,\alpha\right)$. This deformation preserves cell volume. 

In Figure 4 we compare cells for both rigid and fluid spheroids with states at the end of simulation ($t=t_f$) to the same cells strained by $\alpha = 0.5$ according to this volume-preserving deformation. As individual cells are uni-axially strained, we indeed find an increase in maximum shear stress for cells from both fluid-like and solid-like spheroids. For fluid-like cells, the increase in stress is greater as their initial stress is typically smaller (as compared to solid-like cells). Moreover, the maximum shear stress eventually increases nonlinearly with strain.  This is clear evidence for cellular strain-stiffening with the linear stress-strain curve at small strain deviating from linear around 0.1 strain. Moreover, a second estimate of 0.4 strain characterizing the crossover from lower strain to higher strain behavior is found (See inset of Fig. 3c). Intriguingly, strain stiffening has been observed in numerous cell stretching experiments~\cite{Fernandez2008,IcardArcizet2008,Kollmannsberger2011,Zink2012} and even in cell compression studies~\cite{Gandikota2020}. Moreover, collagen matrices also exhibit strain stiffening~\cite{Storm2005,Cheung20243D}. We now observe this phenomenon in an individual vertex model cell long harkening back to the statement that vertex model are, without the rearrangements, many body spring networks beyond the usual three-body angular interaction found in fiber networks and so we expect to observe this phenomenon even at the single cell level with its rather minimal multi-body spring network~\cite{Bi2015}.  This strain stiffening can be enhanced by internal cytoskeletal machinery being developed in the form of acto-myosin stress fibers that span the extent of the cell and can help drive mesenchymal motion along fibers~\cite{Lopez2014,Mayett2017}.

Cells along the boundary of the spheroid can much more readily extend themselves to strain-stiffening than cells in the bulk of the spheroid. If a cell sufficiently strains itself to at least match the local stiffness of the fibers it is interacting with, combined with a decrease in cell-cell adhesion, then it can readily break move out from the spheroid by crawling along the collagen fibers and remodel them as it does so.

\subsection{Single cell versus multi-cell breakout and an extended vertex model }

From a collective perspective, prior work has viewed cell breakout as an unjamming phenomenon~\cite{oswald2017,gottheil2023,Parker2020,Parker2024,cai2025}. More recent work has broken down the unjamming phenomenon into different modes in terms of an active fluid-like unjamming for multiple cells breaking out or an unjamming directly to an active gas-like phase in terms of single cell breakout\cite{Chepizhko2021}. Another study distinguishes between an active fluid and an active nematic, where the cells are aligned as they flow~\cite{Ilina2020}. Parallel cell streams are a cellular realization of an active nematic, though in the MEFs the cellular streams appear to be individual cell streams (Fig. \ref{fig:cell-breakout-results}(a))~\cite{Ho2024}. There may be other types of cell breakout modes~\cite{Friedl2000,NguyenNgoc2012}. The spheroid environment and the cell-cell adhesion both together determine with breakout mode. For instance, cell types with high cell-cell adhesion and low compressive stress from the environment drives collective cell breakout, as opposed to individual cell breakout.

Combining our cellular strain stiffening finding and the observation of either single cell or single-file cell streaming in the MEF spheroids, we ask the
question: {\it What are the minimal mechanical ingredients driving different cell breakout modes?} To answer this question, we introduce a simplified “extended” vertex model in which cell–cell interactions are represented by tunable adhesion springs. This model is not intended to resolve molecular adhesion dynamics, but rather to test how effective adhesion strength and anisotropy influence mechanically driven escape from the spheroid. In this  extended vertex model, cells no longer share interfaces but instead interact via explicit, tunable cell–cell adhesion springs. The detailed construction of this extended model, including cell-cell adhesion springs with stiffness $k_{cc}$, spheroid initialization, and fiber connections, is given in the Supplemental Material. Given now explicit cell-cell adhesions with this extended vertex model, there is prior work studying anisotropy in cell-cell adhesion. For instance, E-cadherin complexes recruit mechanosensitive proteins under tensile loading, leading to polarized adhesion reinforcement~\cite{leDuc2010}. And in mesenchymal or partially mesenchymal cells, stress fibers align along the long axis, and cadherin-mediated adhesions preferentially stabilize at the ends of those fibers~\cite{Mertz2013,Engl2014}. Motivated by these anisotropic cell-cell adhesion findings, we will explore its implications for different types of cell breakout modes.

For a range of $k_{cc}$s, we find that fibers pulling on one boundary cell does not lead to single cell breakout. However, decreasing $k_{cc}$ for the cell-cell adhesion springs associated with the boundary cell being pulled does result in single cell breakout. Specifically, as $k_{cc}$ associated with the boundary cell is decreased 10-fold for every 5 percent strain decrease in the fiber target spring length, we observe the boundary cell inner edge clearly moving away from the closest center cell edge. This reduction represents a coarse-grained weakening of effective adhesion under strain, allowing us to test when intercellular cohesion yields to fiber-mediated forces as the cell is being pulled out. For now, we have decreased all of the associated boundary cell cell-cell adhesion springs and not yet incorporated the anisotropy in cell-cell adhesion. The reduction allows for the pulling force of the fibers to dominate over the cell-cell adhesion, as expected. Obviously, this crossover depends on the ratio of the pulling force to the cell-cell adhesion forces in the direction opposing the pulling. Moreover, the internal floppiness of the cell itself also contributes. Should the boundary cell remain floppy, then it can absorb the pulling force without itself moving away from the other cells in the spheroid by simply deforming with little energy cost. However, given the strain-induced stress transition, eventually the cell tensions and so begins to move away from the other cells. In other words, the fibers pulling, the boundary cell, and the other cell-cell adhesion springs can each be represented each as a spring, connected in series and being pulled in two opposing directions. Of course, given the internal cytoskeletal components of the cell, again, the cell drives its own stress transition by assembling stress fibers more prominently along the long axis of the cell to tension as discussed in the previous subsection and the boundary cell stiffens and ultimately moves away from the spheroid should the fiber pulling force dominate. 


\begin{figure*}[htbp]
    \centering
        \includegraphics[width=\textwidth]{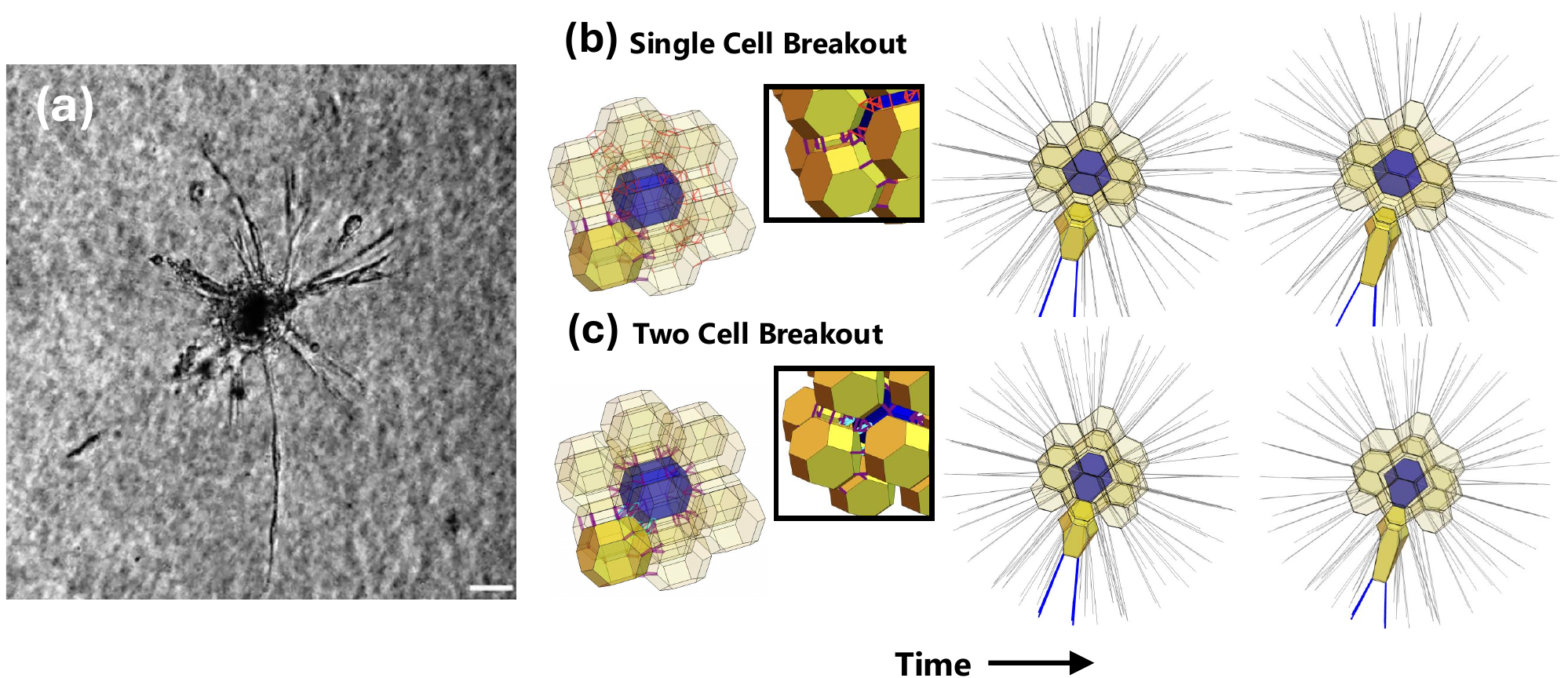}
    \captionsetup{justification=raggedright,singlelinecheck=false}
    \caption{\textit{ Different spheroid cell breakout modes.} (a) Experiment with MEF spheroid in 1.5 mg/ml collagen I. matrix. (b) Single-cell breakout as the target spring length of the four fibers attached to the deep yellow cell decreases. The thicker purple cell-cell adhesions are weaker than the thinner red ones. (c) Two-cell breakout as the target spring length of the four fibers attached to the deep yellow cell decreases. The cyan denotes strong anisotropic cell-cell adhesion springs between the leader boundary dark yellow cell and the center blue cell that is becoming elongated in the dark yellow cell extension to ultimately follow it.  }
    \small\raggedright
    \label{fig:cell-breakout-results}
\end{figure*}

How does the anisotropic cell-cell adhesion come in play? Given our mechanism for single cell breakout, to study the potential for two-cell breakout, we decrease the cell-cell adhesion for both the boundary cell and the center cell as the boundary cell is being pulled by the associated fibers. Interestingly, we do not observe two-cell breakout in this scenario as the boundary cell pulls away from the center cell as before. However, as cells become strained, and thereby exhibiting stress transitions, we therefore introduce a phenomenological rule in which effective cell–cell adhesion is strengthened along the principal elongation axis of a strained cell and weakened orthogonally. This rule captures the mechanical effect of polarized adhesion reinforcement without explicitly modeling molecular adhesion dynamics. In other words, more stress fibers along the long axis of the cell allow for more focal adhesions for the part of the cell in contact with the fibers.  In order to stabilize the cell stress fibers, there needs to be anchoring on the other side of the cell that remains in contact with the spheroid cells. We argue that as cells undergo strain-inducing stress transitions their cell-cell adhesions become weaker in the direction perpendicular to the long cell axis and strengthened in the direction parallel to the long cell axis, as has been found in individual mesenchymal cells. 

In the single cell breakout scenario, there is not as significant a stress transition so that the anisotropic cell-cell adhesion effect is not as dominant. However, in the two cell breakout scenario, there is a more significant stress transition and so this predicted anisotropic cell-cell adhesion effect is more pronounced. When this adhesion anisotropy is introduced, we observe stable two-cell breakout in which a follower blue cell trails the dark yellow leader. This demonstrates that directional adhesion reinforcement is a sufficient mechanical ingredient for streaming-like invasion in this minimal model. Given our results, we predict that the cells in the stream are more strained than single breakout cells, which distinguishes between the two modes of breakout.

\section{Discussion}

Fluid-like and solid-like spheroids differ not simply in how much stress they store internally, but in how stresses are exchanged with the surrounding extracellular matrix. Fluid-like spheroids are the most effective at remodeling collagen as stress relaxation through cell rearrangements drives spheroid shape fluctuations that enable repeated mechanical coupling and emergent reciprocity with the fiber network, facilitating outward stress transmission rather than internal stress storage~\cite{zhang2025}. Solid-like spheroids, in contrast, develop pronounced radial gradients in maximum shear stress, with elevated stresses localized near the boundary and lower stresses in the interior with minimal pathways for stress relaxation within the spheroid. While this organization reflects a mechanically rigid spheroid architecture, it primarily corresponds to stress storage within the spheroid rather than efficient stress transfer to the surrounding fibers. As a result, solid-like spheroids remain comparatively limited in their ability to reorganize either cell–cell contacts or the extracellular matrix.

This distinction reframes a central paradox in embedded spheroid invasion. The spheroids that most strongly remodel the extracellular matrix and invade most efficiently are those in which cells occupy low-stress states on average and lack persistent spatial stress gradients. In principle, mesenchymal invasion requires cells to exert forces comparable to the stiffness of the surrounding collagen fibers in order to engage and remodel them, raising the question of how fluid-like tissues that rapidly dissipate internal stress nonetheless generate the forces required for invasion. Resolving this paradox requires shifting attention away from bulk-averaged stress measures and toward localized, transient stress states that emerge dynamically at the spheroid–matrix interface.

Our results show that this apparent contradiction is resolved at the level of individual cells through strain-induced stress amplification. Although fluid-like spheroids exhibit narrow distributions of maximum shear stress when averaged over the population, individual cells at the spheroid boundary can undergo large, nonlinear increases in stress when sufficiently elongated. This strain stiffening enables cells that are otherwise mechanically soft to transiently access high-stress states, allowing them to exert forces comparable to those of the surrounding collagen fibers. In this way, low-stress cells can locally become high-stress cells precisely when geometric constraints and matrix interactions demand it. Force generation is therefore shifted away from sustained stress accumulation in the bulk and concentrated instead in dynamically strained boundary cells, whose stress states are short-lived but mechanically consequential. Invasion emerges not from static stress localization, but from repeated, geometry-driven stress transitions triggered by cell elongation during migration.

Extending the formulation of the cell stress tensor to three-dimensional polyhedral cells further clarifies how geometry and mechanics interact in this process. In solid-like spheroids, cell shape anisotropy and stress are not generically aligned, indicating that elongation alone is an unreliable indicator of mechanical state. In fluid-like spheroids, however, the rare high-stress cells exhibit preferential alignment between elongation and maximal stress directions, consistent with strain-induced stress amplification. This decoupling of shape and stress highlights the limitations of using cell morphology as a proxy for force generation, particularly in three-dimensional tissues where geometric and mechanical constraints are more complex. Instead, it underscores the need for direct, stress-based metrics when interpreting spheroid mechanics and invasion potential. The three-dimensional stress formulation introduced here thus provides a quantitative bridge between vertex-model descriptions of tissue mechanics and experimentally accessible observables such as traction forces and collagen fiber deformations.

This framework also reveals that different invasion modes arise from distinct mechanical pathways. Single-cell breakout requires both cell strain stiffening and a local reduction in cell–cell adhesion, allowing fiber-mediated forces to overcome intercellular cohesion and permit individual cells to escape the spheroid. In contrast, collective streaming invasion cannot be achieved through uniform adhesion weakening alone. Instead, it relies on anisotropic cell–cell adhesion, with adhesion strengthened along the long axis of strained, high-stress cells and weakened in the orthogonal direction. This anisotropy stabilizes follower cells behind a leader, enabling coordinated motion that is reminiscent of active nematic behavior. As a result, streaming invasion is predicted to involve cells that are, on average, more strained and exhibit stronger adhesion anisotropy than those undergoing isolated breakout, providing a clear experimental distinction between invasion modes.

Taken together, our results suggest that invasive behavior in embedded spheroids is governed by a hierarchy of mechanical processes spanning scales: spheroid-level rheology controls stress distributions and overall fiber network remodeling capacity; cell-level strain stiffening enables more local force generation; and adhesion anisotropy selects invasion mode. As we move to more detailed descriptions of how unjamming emerges, we must offer increasingly more testable predictions for experiments, including anisotropic cadherin reinforcement, elevated cellular strain in streaming cells relative to isolated invaders, and even bringing into the picture the cell nucleus~\cite{Liu2021,Berg2023,Zhang2026} as well as the intelligence of cells themselves~\cite{Anisetti2023,Anisetti2022,Ameen2025} to disentangle the multiscale details of cellular invasion.

\section*{Funding}
Tao Zhang acknowledges financial support from the NSFC/China via award 22303051. Mingming Wu and JMS acknowledge financial support from the National Science Foundation under Grant PoLS-2014192. Alison E. Patteson and JMS acknowledge financial support from the National Science Foundation under grant PoLS-2412961. Alison E. Patteson acknowledges NSF CMMI 2238600 and NIH R35GM142963. 

\section*{Acknowledgements}
Shabeeb Ameen acknowledges Gentian Muhaxheri for useful discussion regarding stress tensors.

\section*{Data Availability}
The data and code required to reproduce our analyses are available on Zenodo at https://doi.org/10.5281/zenodo.18514493. 

\bibliography{espheroid3D}

@article{Liu2021,
  title={Dynamic nuclear structure emerges from chromatin cross-links and motors},
  author={Liu, Kuang and Patteson, Alison E and Banigan, Edward J and Schwarz, J. M.},
  journal={Physical Review Letters},
  volume={126},
  number={15},
  pages={158101},
  year={2021},
  publisher={APS}
}

@article{Berg2023,
  title={Transcription inhibition suppresses nuclear blebbing and rupture independently of nuclear rigidity},
  author={Berg, Isabel K and Currey, Marilena L and Gupta, Sarthak and Berrada, Yasmin and Nguyen, Bao V and Pho, Mai and Patteson, Alison E and Schwarz, J. M. and Banigan, Edward J and Stephens, Andrew D},
  journal={Journal of cell science},
  volume={136},
  number={20},
  pages={jcs261547},
  year={2023},
  publisher={The Company of Biologists Ltd}
}

@article{Zhang2026,
  title={How human-derived brain organoids are built differently from brain organoids derived from genetically-close relatives: A multi-scale hypothesis},
  author={Zhang, Tao and Gupta, Sarthak and Lancaster, Madeline and Schwarz, J. M.},
  journal={Soft Matter},
  year={2026},
  publisher={Royal Society of Chemistry}
}

@article{Mayett2017,
  title={Chase-and-run dynamics in cell motility and the molecular rupture of interacting active elastic dimers},
  author={Mayett, David and Bitten, Nicholas and Das, Moumita and Schwarz, J. M.},
  journal={Physical Review E},
  volume={96},
  number={3},
  pages={032407},
  year={2017},
  publisher={APS}
}

@article{Storm2005,
  title={Nonlinear elasticity in biological gels},
  author={Storm, Cornelis and Pastore, Jennifer J and MacKintosh, Fred C and Lubensky, Tom C and Janmey, Paul A},
  journal={Nature},
  volume={435},
  number={7039},
  pages={191--194},
  year={2005},
  publisher={Nature Publishing Group UK London}
}

@article{Gandikota2020,
  title={Loops versus lines and the compression stiffening of cells},
  author={Gandikota, Mahesh C and Pogoda, Katarzyna and Van Oosten, Anne and Engstrom, Tyler A and Patteson, Alison E and Janmey, P. A. and Schwarz, J. M.},
  journal={Soft Matter},
  volume={16},
  number={18},
  pages={4389--4406},
  year={2020},
  publisher={Royal Society of Chemistry}
}

@article{Cheung20243D,
  title={3D traction force microscopy in biological gels: from single cells to multicellular spheroids},
  author={Cheung, Brian CH and Abbed, Rana J and Wu, Mingming and Leggett, Susan E},
  journal={Annual Review of Biomedical Engineering},
  volume={26},
  year={2024},
  publisher={Annual Reviews}
}

@article{Ameen2025,
  title={},
  author={Ameen, Shabeeb and Zhang, Tao and Schwarz, J. M.},
  journal={in preparation},
  volume={},
  number={},
  pages={},
  year={2025},
  publisher={}
}

@article{leDuc2010,
  title={Vinculin potentiates E-cadherin mechanosensing and is recruited to actin-anchored sites within adherens junctions},
  author={le Duc, Q. and Shi, Q. and Blonk, I. and Sonnenberg, A. and Wang, N. and Leckband, D. and de Rooij, J.},
  journal={Journal of Cell Biology},
  volume={189},
  pages={1107--1115},
  year={2010}
}

@article{Engl2014,
  title={Actin dynamics modulate mechanosensitive junction remodeling},
  author={Engl, W. and Arasi, B. and Yap, L. L. and Thiery, J. P. and Viasnoff, V.},
  journal={Journal of Cell Biology},
  volume={207},
  pages={639--654},
  year={2014}
}

@article{Mertz2013,
  title={Cadherin-based intercellular adhesions organize epithelial cell–matrix traction forces},
  author={Mertz, A. F. and Che, Y. and Banerjee, S. and Goldstein, J. M. and Rosowski, K. A. and Revilla, S. F. and Niessen, C. M. and Marchetti, M. C. and Dufresne, E. R. and Horsley, V.},
  journal={Proceedings of the National Academy of Sciences},
  volume={110},
  pages={842--847},
  year={2013}
}

@article{Zink2012,
  title        = {Der optical stretcher — kontaktfreie Messung der Zellmechanik},
  author       = {Zink, Mareike and Braunm{\"u}ller, Sabrina and Nnetu, Kelvin David and Knorr, Melanie and K{\"a}s, Josef A.},
  journal      = {Advances in Experimental Medicine and Biology},
  year         = {2012},
  volume       = {737},
  pages        = {21--45}
}

@article{IcardArcizet2008,
  title        = {Cell stiffening in response to external stress is correlated to actin recruitment},
  author       = {Icard-Arcizet, Delphine and Cardoso, Olivier and Richert, Alain and H{\'e}non, Sylvie},
  journal      = {Biophysical Journal},
  year         = {2008},
  volume       = {94},
  number       = {7},
  pages        = {2906--2913},
  doi          = {10.1529/biophysj.107.118265}
}

@article{Fernandez2008,
  title        = {Single cell mechanics: stress stiffening and kinematic hardening},
  author       = {Fern{\'a}ndez, Pablo and Ott, Albrecht},
  journal      = {Physical Review Letters},
  year         = {2008},
  volume       = {100},
  number       = {23},
  pages        = {238102},
  doi          = {10.1103/PhysRevLett.100.238102}
}

@article{Kollmannsberger2011,
  title        = {Nonlinear viscoelasticity of adherent cells is controlled by cytoskeletal tension},
  author       = {Kollmannsberger, Philip and Mierke, Claudia T. and Fabry, Ben},
  journal      = {Soft Matter},
  year         = {2011},
  volume       = {7},
  number       = {7},
  pages        = {3127--3132},
  doi          = {10.1039/C0SM00833H}
}

@article{cai2025,
  title={Unjamming Transition as a Paradigm for Biomechanical Control of Cancer Metastasis},
  author={Cai, Grace and Rodgers, Nicole C and Liu, Allen P},
  journal={Cytoskeleton},
  volume={82},
  number={6},
  pages={388--403},
  year={2025},
  publisher={Wiley Online Library}
}

@article{gottheil2023,
  title={State of cell unjamming correlates with distant metastasis in cancer patients},
  author={Gottheil, Pablo and Lippoldt, Juergen and Grosser, Steffen and Renner, Frederic and Saibah, Mohamad and Tschodu, Dimitrij and Possoegel, Anne-Kathrin and Wegscheider, Anne-Sophie and Ulm, Bernhard and Friedrichs, Kay and others},
  journal={Physical Review X},
  volume={13},
  number={3},
  pages={031003},
  year={2023},
  publisher={APS}
}

@article{oswald2017,
  title={Jamming transitions in cancer},
  author={Oswald, Linda and Grosser, Steffen and Smith, David M and K{\"a}s, Josef A},
  journal={Journal of physics D: Applied physics},
  volume={50},
  number={48},
  pages={483001},
  year={2017},
  publisher={IOP Publishing}
}

@article{NguyenNgoc2012,
  title={Collective cell invasion of the basement membrane in breast cancer},
  author={Nguyen-Ngoc, Kien-Vinh and Cheung, Kevin J. and Brenot, Antoine and Shamir, Eliah R. and Gray, Rebecca S. and Hines, W. Charles and Yaswen, Paul and Werb, Zena and Ewald, Andrew J. and Weaver, Valerie M.},
  journal={Proceedings of the National Academy of Sciences},
  volume={109},
  number={39},
  pages={E2595--E2604},
  year={2012}
}

@article{Chepizhko2021,
  title={Jamming–unjamming transition in dense cell collectives},
  author={Chepizhko, Oleksandr and Alert, Ricard and Wingreen, Ned S. and Fredberg, Jeffrey J.},
  journal={iScience},
  volume={24},
  pages={102271},
  year={2021}
}

@article{Friedl2000,
  title={The biology of cell locomotion within three-dimensional extracellular matrix},
  author={Friedl, Peter and Br{\"o}cker, Ernst-Bernhard},
  journal={Cellular and Molecular Life Sciences},
  volume={57},
  pages={41--64},
  year={2000}
}

@article{nestor2018,
  title={Relating cell shape and mechanical stress in a spatially disordered epithelium using a vertex-based model},
  author={Nestor-Bergmann, Alexander and Goddard, Georgina and Woolner, Sarah and Jensen, Oliver E},
  journal={Mathematical medicine and biology: a journal of the IMA},
  volume={35},
  number={Supplement\_1},
  pages={i1--i27},
  year={2018},
  publisher={Oxford University Press}
}

@article{zhang2025,
  title={Enhanced extracellular matrix remodeling due to embedded spheroid fluidization},
  author={Zhang, Tao and Ameen, Shabeeb and Ghosh, Sounok and Kim, Kyungeun and Pandey, Mrinal and Cheung, Brian and Thanh, Minh and Patteson, Alison E and Wu, Mingming and Schwarz, J. M.},
  volume={27},
pages = {073301},
  journal={New Journal of Physics},
  year={2025}
}

@article{pandey2025,
  title={Viscoelastic properties of tumor spheroids revealed by a microfluidic compression device and a modified power law model},
  author={Pandey, Mrinal and Zhu, Bangguo and Roach, Kaitlyn and Suh, Young Joon and Segall, Jeffrey E and Hui, Chung-Yuen and Wu, Mingming},
  journal={arXiv preprint arXiv:2509.17294},
  year={2025}
}

@article{kim2007,
  title={A hybrid model for tumor spheroid growth in vitro I: theoretical development and early results},
  author={Kim, Yangjin and Stolarska, Magdalena A and Othmer, Hans G},
  journal={Mathematical Models and Methods in Applied Sciences},
  volume={17},
  number={supp01},
  pages={1773--1798},
  year={2007},
  publisher={World Scientific}
}

@article{friedl1998,
  title={Cell migration strategies in 3-D extracellular matrix: Differences in morphology, cell matrix interactions, and integrin function},
  author={Friedl, Peter and Z{\"a}nker, Kurt S and Br{\"o}cker, Eva-B},
  journal={Microscopy research and technique},
  volume={43},
  number={5},
  pages={369--378},
  year={1998},
  publisher={Wiley Online Library}
}

@article{Pandey2024,
  author={Pandey, M. and Suh, Y. J. and Kim, M. and Davis, H. J. and Segall, J. E. and Wu, M. },
  title={Mechanical compression regulates tumor spheroid invasion into a 3D collagen matrix. },
  journal={Physical Biology},
  volume={21},
  number={3},
  pages={036003},
  year={2024},
  publisher={Elsevier}
}

@article{Ho2024,
  title={Vimentin promotes collective cell migration through collagen networks via increased matrix remodeling and spheroid fluidity},
  author={Ho Thanh, Minh Tri and Poudel, Arun and Ameen, Shabeeb and Carroll, Robert and Wu, Mingming and Soman, Pranav and Zhang, Tao and Schwarz, J. M. and Patteson, Alison E.},
  journal={bioRxiv},
  pages={2024--06},
  year={2024},
  publisher={Cold Spring Harbor Laboratory}
}

@article{Parker2024,
  title={},
  author={Amanda Parker and J. M. Schwarz},
  journal={arXiv:2503.14319 },
  volume={},
  number={},
  pages={},
  year={2025},
  publisher={}
}

@article{Anisetti2023,
  title={Learning by non-interfering feedback chemical signaling in physical networks},
  author={Vidyesh Rao Anisetti and Benjamin Scellier and J. M. Schwarz},
  journal={Physical Review Research},
  volume={5},
  number={2},
  pages={023024},
  year={2023},
  publisher={APS}
}

@article{Ilina2020,
  title={Cell--cell adhesion and 3D matrix confinement determine jamming transitions in breast cancer invasion},
  author={Ilina, Olga and Gritsenko, Pavlo G and Syga, Simon and Lippoldt, J{\"u}rgen and La Porta, Caterina AM and Chepizhko, Oleksandr and Grosser, Steffen and Vullings, Manon and Bakker, Gert-Jan and Starru{\ss}, J{\"o}rn and others},
  journal={Nature Cell Biology},
  volume={22},
  number={9},
  pages={1103--1115},
  year={2020},
  publisher={Nature Publishing Group UK London}
}

@article{Lopez2014,
  title={Active elastic dimers: Cells moving on rigid tracks},
  author={J. H. Lopez and Moumita Das and J. M. Schwarz},
  journal={Physical Review E},
  volume={90},
  number={3},
  pages={032707},
  year={2014},
  publisher={APS}
}

@article{Parker2020,
  title={How does the extracellular matrix affect the rigidity of an embedded spheroid?},
  author={Parker, Amanda and Marchetti, M Cristina and Manning, M Lisa and Schwarz, JM},
  journal={Physical Review E},
  volume={111},
  number={4},
  pages={044410},
  year={2025},
  publisher={APS}
}

@article{Zhang2022,
  title={Topologically-protected interior for three-dimensional confluent cellular collectives},
  author={Zhang, Tao and Schwarz, J. M.},
  journal={Physical Review Research},
  volume={4},
  number={4},
  pages={043148},
  year={2022},
  publisher={APS}
}

@article{Huang2020,
  title={Tumor spheroids under perfusion within a 3D microfluidic platform reveal critical roles of cell-cell adhesion in tumor invasion},
  author={Huang, Yu Ling and Ma, Yujie and Wu, Cindy and Shiau, Carina and Segall, Jeffrey E and Wu, Mingming},
  journal={Scientific reports},
  volume={10},
  number={1},
  pages={9648},
  year={2020},
  publisher={Nature Publishing Group UK London}
}

@Article{Bi_2014,
author ="Bi, Dapeng and Lopez, Jorge H. and Schwarz, J. M. and Manning, M. Lisa",
title  ="Energy barriers and cell migration in densely packed tissues",
journal  ="Soft Matter",
year  ="2014",
volume  ="10",
issue  ="12",
pages  ="1885-1890",
publisher  ="The Royal Society of Chemistry",
doi  ="10.1039/C3SM52893F",
url  ="http://dx.doi.org/10.1039/C3SM52893F"
}

@article{Lui_2017,
    author = {Liu, A.P. and Chaudhuri, O. and Parekh, S.H.},
    title = "{New advances in probing cell–extracellular matrix interactions}",
    journal = {Integrative Biology},
    volume = {9},
    number = {5},
    pages = {383-405},
    year = {2017},
    month = {03},
    issn = {1757-9708},

}

@article{Bi2015,
author = {Bi, Dapeng and Lopez, J. H. and Schwarz, J. M. and Manning, M. Lisa},
doi = {10.1038/nphys3471},
issn = {17452481},
journal = {Nature Physics},
number = {12},
pages = {1074--1079},
title = {{A density-independent rigidity transition in biological tissues}},
volume = {11},
year = {2015}
}

\clearpage
\onecolumngrid
\section*{Supplemental Material}
\setcounter{figure}{0}
\renewcommand{\thefigure}{S\arabic{figure}}

\subsection{3D Computational Model}

 The spheroid is represented using a vertex model in which cells are deformable polyhedra governed by quadratic volume and surface area constraints, along with interfacial tension at the spheroid boundary, as encoded in the spheroid energy functional 
\begin{equation}
\label{eq:energy}
E_{SP}= K_V\sum_{j}(V_{j} - V_{0})^2+K_A\sum_{j}(A_{j} -A_{0})^2 +\Gamma\sum_{\alpha}\delta_{\alpha,B} S_{\alpha},
\end{equation}
where $A_{j}$ and $V_{j}$ represent the area and volume of the $j$th cell, respectively.  The fibrous matrix is modeled as a diluted face-centered cubic (FCC) lattice of crosslinked springs with bending stiffness, forming a disordered network that mimics collagen and encoded in the fiber network energy functional
\begin{equation}
E_{FB}=\frac{K_{\text{S}}}{2}\sum_{<ij>} f_{ij}\;(l_{ij}-l_0)^2+\\ 
\frac{K_{\text{B}}}{2}\sum_{m=1}^{3}\sum_{<ijk>=\pi}f_{ij}f_{jk}\;(\theta_{ijk}^m-\pi)^2,
\end{equation}
where $f_{ij}=1$ if a bond is occupied and $f_{ij}=0$ if not, $l_{ij}$ represents the length of each bond, $\sum_{\langle ij \rangle}$ represents sum over all nearest neighbor bonds, $\theta_{ijk}^m$ represents the angle between nearest neighbor bonds that are crosslinked by the same phantom node with phantom node index $m$. Moreover, $\sum_{\langle ijk \rangle=\pi_0}$ represents sum over pairs of nearest neighbor bonds sharing a node and only for those aligned along of the principle axes of the initial FCC lattice. To couple the spheroid to the matrix, a fixed number of active linker springs dynamically attach fiber nodes to the center of boundary cell surfaces and contract over time, simulating actomyosin contractility mediated by focal adhesions. Specifically, the total energy of the active linker springs is quantified by 
\begin{equation}
    E_{LS}=\frac{K_{LS}}{2}\sum_{i=1}^{N_{ls}} (l'_i-l'_{0i}(t))^2,
\end{equation}
where $K_{LS}$ represents the stiffness of each linker spring and $N_{ls}$, the total number of linker springs. 

The total energy of the coupled systems is the sum of these individual energies. To evolve the coupled system, we use overdamped Brownian dynamics for the cell vertices and Euler-Muruyama integration for the fiber nodes. Reconnection events within the spheroid implemented throughout the simulation. Simulations begin with a spheroid inserted into the center of a pruned fiber network, and linker springs gradually shorten to simulate contractility. This model allows exploration of how cell shape (via shape index), matrix stiffness (via bond occupation probability), and now cell stress impacts mechanical and morphological evolution of the embedded spheroid system. Please see the table for the simulation parameters used. 

\begin{table*}
\begin{center}
\begin{tabular} { l | c | c }
{\bf Quantity} & {\bf Symbol} & {\bf Value }\\
 
 Simulation timestep & $\frac{dt}{t_0}$ & $0.005$ \\   
 Simulation time in total & $\frac{t_f}{t_0}$ & $25000.$ \\
 Cell area stiffness & $\frac{K_A}{K_V\,V_0^{2/3}}$ & $0.1$ \\   
 Cell target surface area & $s_0$ & 5.2,5.8 \\  
 Interfacial tension & $\frac{\Gamma}{K_V\,V_0^{4/3}}$ & 1.0 \\
 Reconnect. Event threshold edge length & $\frac{l_{th}}{V_0^{1/3}}$ & $0.02$ \\
 Damping & $\frac{K_V\,V_0^{4/3}\,t_0}{\mu}$ & $1$\\   
 Active force fluctuation energy & $\frac{k_BT_{eff}}{K_V\,V_0^2}$ & $10^{-4}$ \\
 Individual fiber bending stiffness & $\frac{K_B}{K_S l'_{0}(0)^{2}}$ &  $10^{-4}$\\
 Individual fiber diameter & $\frac{d_f}{l'_0(0)}$& 0.02\\
 Fiber network pore size & $\frac{\xi}{l'_0(0)}$ & 1.0\\
 Edge occupation probability for fiber network & $p$ & $0.75-0.95$\\
 Active linker spring stiffness & $\frac{K_{LS}}{K_S \,l_0(0)'^2}$ & 1.0\\
 Final active linker spring target length & $\frac{l'_{t0}}{l'_0(0)}$ & 0.2\\
Number of spheroid cells & $N$ & $~400$ \\
 Number of active linker springs & $N_{LS}$ & 100\\
 Number of fiber network nodes & $N_f$ & $~1000$\\
 Number of realizations & $N_R$ & $30$\\
\end{tabular}
\caption{Table of the dimensionless parameters used in the embedded spheroid simulations.}
\end{center}
\end{table*}
\subsection{Cell Properties: Stress and Shape}

We now introduce a generalizable, triangulated cell stress tensor calculation and provide physical and geometrical intuition for the stress components. We will also provide analytical intuition for the maximum shear stress component obtained from the stress tensor analysis. We also address the gyration tensor for a cell and the accompanying anisotropy measure. 

\subsubsection{3D Cell stress tensor}

The energy function for a single polyhedron with triangular faces is given by:
\begin{align}
E = k_V\bracket{V(\{\vect{r_\mu}\}) - V_0}^2 
+ k_S\bracket{S(\{\vect{r_\mu}\}) - S_0}^2,
\end{align}
where $V(\{\vect{r_\mu}\})$, the volume of the polyhedron and $S(\{\vect{r_\mu}\})$, the surface area of the polyhedron, are both  functions of $\{\vect{r_\mu}\}$, the set of position vectors of polyhedron vertices; and $V_0$ are the $S_0$ are the target volume and target surface area, respectively.

Let us compute the Cauchy stress tensor derived from the forces generated by this energy function. The force on vertex $\mu$ is given by the negative gradient of the energy with respect to vertex position $\vect{r_\mu}$:
\begin{align}\label{force}
\vect{F_\mu} & = -2k_V(V-V_0)\frac{\partial V}{\partial \vect{r_\mu}} 
- 2k_S(S-S_0)\frac{\partial S}{\partial \vect{r_\mu}}.
\end{align}
The Cauchy stress tensor for the single polyhedral cell is then defined as:
\begin{align}\label{cauchy_stress_tensor}
\vect{\sigma} = \frac{1}{V}\sum_\mu \vect{r_\mu} \otimes \vect{F_\mu}.
\end{align}
Substituting (\ref{force}) into (\ref{cauchy_stress_tensor}) leads to 
\begin{align}\label{cauchy_stress_tensor_vertex_model}
    \vect{\sigma} =& 
    -\frac{2k_V(V-V_0)}{V}\sum_\mu \vect{r_\mu} \otimes \frac{\partial V}{\partial \vect{r_\mu}} \nonumber\\
    &- \frac{2k_S(S-S_0)}{V}\sum_\mu \vect{r_\mu} \otimes \frac{\partial S}{\partial \vect{r_\mu}}.
\end{align}

The volume of the polyhedron can be expressed as a sum of (signed) tetrahedral volumes, i.e.,
\begin{align}\label{volume}
    V = \frac{1}{3}\sum_{\mathcal{F}} \vect{S}_\mathcal{F} \cdot \vect{d}_\mathcal{F},
\end{align}
where $\vect{d}_\mathcal{F}$ is the vector from some arbitrary origin $\vect{O}$ to any point on the plane on triangular face $\mathcal{F}$ of the polyhedron; and $\vect{S}_\mathcal{F}$ is the outward-oriented area vector of $\mathcal{F}$, or 
\begin{align}\label{triangle_area}
    \vect{S}_\mathcal{F} 
    &= \frac{1}{2}\sum_{\mu \in \mathcal{F}} 
    (\vect{r_\mu} - \vect{d}_\mathcal{F}) \times 
    (\vect{r_{\mu+1}} - \vect{d}_\mathcal{F})\nonumber\\
    &= \frac{1}{2}\sum_{\mu \in \mathcal{F}} 
    \vect{r_\mu}\times \vect{r_{\mu+1}},
\end{align}
where the summation over the vertices of $\mathcal{F}$ runs counter-clockwise (with respect to outward orientation). We note that the simplification (\ref{triangle_area}) follows from the zero-sum of the edge vectors of any face. Notably, the surface area of the polyhedron is then the sum of the areas of all the triangular faces:
\begin{align}\label{surface_area}
S = \sum_f \norm{\vect{S}_\mathcal{F}}.
\end{align}

Now, for any polyhedron vertex $\mu$, the derivative of polyhedral volume V with respect to vertex position $\vect{r_\mu}$ involves only the faces that contain that vertex:

\begin{align}\label{volume_derivative}
    \frac{\partial V}{\partial \vect{r_\mu}} 
    &= 
    \frac{1}{3}\sum_{\mathcal{F} \ni \mu} 
    \left[ 
        \frac{\partial \vect{S}_\mathcal{F}}{\partial \vect{r_\mu}} \cdot \vect{d}_\mathcal{F} 
        + \vect{S}_\mathcal{F} \cdot \frac{\partial \vect{d}_\mathcal{F}}{\partial \vect{r_\mu}} 
    \right].
\end{align}

The same applies to the derivative of the polyhedral surface area S with respect to vertex position $\vect{r_\mu}$:
\begin{align}\label{surface_derivative}
    \frac{\partial S}{\partial \vect{r_\mu}} 
    &= \sum_{\mathcal{F} \ni \mu} 
    \frac{\partial \norm{\vect{S}_\mathcal{F}}}{\partial \vect{r_\mu}}.
\end{align}
We insert (\ref{volume_derivative}) and (\ref{surface_derivative}) into (\ref{cauchy_stress_tensor_vertex_model}), noting that it is possible to interchange the order of summations, such that
\begin{align*}
    \vect{\sigma} =& 
    -\frac{2k_V(V-V_0)}{V}\\
    &\quad\quad
    \sum_\mathcal{F} \sum_{\mu\in \mathcal{F}}
    \vect{r_\mu} \otimes \frac{1}{3}
    \bracket{
    \frac{\partial \vect{S}_\mathcal{F}}{\partial \vect{r_\mu}} \cdot \vect{d}_\mathcal{F} 
    + \vect{S}_\mathcal{F} \cdot \frac{\partial \vect{d}_\mathcal{F}}{\partial \vect{r_\mu}}
    }
    \\
    &- \frac{2k_S(S-S_0)}{V}
    \sum_\mathcal{F} \sum_{\mu\in \mathcal{F}}
    \vect{r_\mu} 
    \otimes 
    \frac{\partial \norm{\vect{S}_\mathcal{F}}}{\partial \vect{r_\mu}}.
\end{align*}
In Sec.~\ref{appendix_identities} of this Supplemental Material, we develop a series of identities that further simplify the stress tensor. In particular, Eqs.~(\ref{identity_tensor_product_1}), (\ref{identity_tensor_product_2}), and (\ref{identity_tensor_product_3}) reduce the above expression to the following form for the stress tensor for a single polyhedron:

\begin{align}\label{stress_tensor_final}
    \vect{\sigma} 
    =& 
    -2k_V(V-V_0)\mathbb{I}
    - \frac{2k_S(S-S_0)}{V}
    S \mathbb{I}
    \nonumber\\
    &
    + \frac{2k_S(S-S_0)}{V}
    \sum_f 
    \frac{1}{\norm{\vect{S}_\mathcal{F}}}
    \bracket{
    \vect{S}_\mathcal{F} 
    \otimes 
    \vect{S}_\mathcal{F}
    }.
\end{align}

Let $\sigma_1$, $\sigma_2$, $\sigma_3$ be the eigenvalues of \ref{stress_tensor_final} in ascending order, corresponding to normalized eigenvectors $\vect{\sigma_1}$, $\vect{\sigma_2}$, $\vect{\sigma_3}$. The maximum shear stress is then defined in terms of the difference between the maximum and minimum principal stresses, namely
\begin{align}\label{maximum_shear_stress}
    \sigma_{shear} = \frac{1}{2}(\sigma_3-\sigma_1).
\end{align}

\subsubsection{3D Cell gyration tensor}

Given $\{\vect{r_\mu}\}$, the shape of the polyhedral cell can be described by the corresponding gyration tensor, 
\begin{align}\label{shape_tensor}
    \vect{G} 
    =&
    \frac{1}{N}
    \sum_{\mu}
    \vect{r}_\mu
    \otimes 
    \vect{r}_\mu,
\end{align}
where N is the number of vertices. Let $g_1$, $g_2$, $g_3$ be the eigenvalues of \ref{shape_tensor} in ascending order, corresponding to normalized eigenvectors $\vect{g_1}$, $\vect{g_2}$, $\vect{g_3}$. The relative shape anisotropy $\kappa^2$ described by the distribution of the vertices is then

\begin{align}
    \kappa^2 =& 
    \frac{3}{2}
    \frac{g_1^2+g_2^2+g_3^2}
    {\left(g_1+g_2+g_3\right)^2}
    -
    \frac{1}{2}.
\end{align}
The relative shape anisotropy, bounded between 0 and 1, represents the elongation of the cell, with $\kappa=0$ corresponding to a spherical distribution of the vertices and $\kappa=1$ representing a linear distribution. The axis of elongation is captured by the eigenvector $\vect{g_3}$.

\subsection{Identities used in the stress tensor derivation}
\label{appendix_identities}

\newcounter{idCounter}
\newcounter{coroCounter}
\setcounter{idCounter}{0}
\newtheorem{identityTemp}{Identity}
\renewcommand{\theidentityTemp}{\arabic{idCounter}}
\newenvironment{identity}{
    \setcounter{coroCounter}{0}
    \stepcounter{idCounter}
    \begin{identityTemp}
}{
    \end{identityTemp}
    
}
\newtheorem{corollaryTemp}{Corollary}[identityTemp]
\renewcommand{\thecorollaryTemp}{\arabic{idCounter}.\arabic{coroCounter}}
\newenvironment{corollary}{
    \stepcounter{coroCounter}
    \begin{corollaryTemp}
}{
    \end{corollaryTemp}
}
The following is a collection of identities useful in simplifying the analytical expression of the stress tensor.

Conventions: Lower case Greek indices for vertex labels, lower case Latin indices for vector components. Twice-repeated lower case Latin indices within any term indicates Einstein summation.

\begin{identity}
    \begin{align} \label{identity_derivative_triangular_area_1}
        \frac{\partial \vect{S}_\mathcal{F}}{\partial \vect{r_\mu}} \cdot \vect{d}_\mathcal{F} 
        &= \frac{1}{2}
        (\vect{r_{\mu+1}} - \vect{r_{\mu-1}}) \times \vect{d}_\mathcal{F}
    \end{align}
\end{identity}

\begin{proof}
Expressing L.H.S. in component form and plugging in (\ref{triangle_area}):
\begin{align*}
    \frac
        {\partial \bracket{\vect{S}_\mathcal{F}}^j }
        {\partial \bracket{\vect{r_\mu}}^i}
    \bracket{\vect{d}_\mathcal{F}}^j
    &=
    \bracket{
    \sum_{\nu \in \mathcal{F}}
    \frac{1}{2}
    \frac
        {\partial}
        {\partial \bracket{\vect{r_\mu}}^i}
        \bracket{
        \vect{r_\nu}
        \times
        \vect{r_{\nu+1} }
        }^j
    }
    \bracket{\vect{d}_\mathcal{F}}^j\\
    &=
    \frac{1}{2}
    \epsilon_{jim}
    \bracket{\vect{r_{\mu+1}} - \vect{r_{\mu-1}}}^m
    \bracket{\vect{d}_\mathcal{F}}^j
\end{align*}
which is R.H.S in component form.
\end{proof}

\begin{corollary}
\begin{align}\label{identity_derivative_triangular_area_2}
    \frac{\partial \norm{\vect{S}_\mathcal{F}}}{\partial \vect{r_\mu}} 
    &= 
    \frac{1}{2}
    \bracket{\vect{r_{\mu+1}} - \vect{r_{\mu-1}}} 
    \times 
    \frac{\vect{S}_\mathcal{F}}{\norm{\vect{S}_\mathcal{F}}}
\end{align}
\end{corollary}

\begin{proof}
    \begin{align*}
        \frac{\partial \norm{\vect{S}_\mathcal{F}}}{\partial \vect{r_\mu}} 
        &= 
        \frac{1}{\norm{\vect{S}_\mathcal{F}}}
        \frac{\partial \vect{S}_\mathcal{F}}{\partial \vect{r_\mu}}
        \cdot \vect{S}_\mathcal{F}
    \end{align*}
which, following the proof in (\ref{identity_derivative_triangular_area_1}), simplifies as required.
\end{proof}

\begin{identity}   
    \begin{align}\label{identity_tensor_product_1}
        \sum_{\mu \in \mathcal{F}} 
        \vect{r_\mu}
        \otimes
        \frac{\partial \vect{S}_\mathcal{F}}{\partial \vect{r_\mu}}
        \cdot
        \vect{d}_\mathcal{F}
        &= 
        \bracket{
            \vect{S}_\mathcal{F}
            \cdot
            \vect{d}_\mathcal{F}
        } \mathbb{I}
        - \vect{d}_\mathcal{F}
        \otimes
        \vect{S}_\mathcal{F}
    \end{align}
\end{identity}

\begin{proof}
Expressing L.H.S. in component form and plugging in (\ref{identity_derivative_triangular_area_1}) 
\begin{align*} 
    &\sum_{\mu \in \mathcal{F}} 
    \bracket{\vect{r_\mu}}^i
    \frac
        {\partial \bracket{\vect{S}_\mathcal{F}}^k }
        {\partial \bracket{\vect{r_\mu}}^j}
    \bracket{\vect{d}_\mathcal{F}}^k\\
    =& 
    \sum_{\mu \in \mathcal{F}} 
    \frac{1}{2}
    \bracket{\vect{r_{\mu}}}^i
    \epsilon_{kjm}
    \bracket{\vect{r_{\mu+1}} - \vect{r_{\mu-1}}}^m
    \bracket{\vect{d}_\mathcal{F}}^k\\
    =&
    \sum_{\mu \in \mathcal{F}} 
    \frac{1}{2}
    \epsilon_{kjm}
    \bracket{
        \delta_{ip}\delta_{mq}
        -\delta_{iq}\delta_{mp}
    }
    \bracket{\vect{r_{\mu}}}^p
    \bracket{\vect{r_{\mu+1}}}^q
    \bracket{\vect{d}_\mathcal{F}}^k\\
    =&
    \sum_{\mu \in \mathcal{F}} 
    \frac{1}{2}
    \epsilon_{kjm}
    \epsilon_{sim}
    \epsilon_{spq}
    \bracket{\vect{r_{\mu}}}^p
    \bracket{\vect{r_{\mu+1}}}^q
    \bracket{\vect{d}_\mathcal{F}}^k\\
    =&
    \delta_{ij}
    \bracket{\vect{S}_\mathcal{F}}^k
    \bracket{\vect{d}_\mathcal{F}}^k
    -
    \bracket{\vect{d}_\mathcal{F}}^i
    \bracket{\vect{S}_\mathcal{F}}^j
\end{align*}
which is R.H.S. in component form.
\end{proof}

\begin{corollary}
    \begin{align}\label{identity_tensor_product_2}
    \sum_{\mu\in \mathcal{F}}
    \vect{r_\mu} 
    \otimes 
    \frac{\partial \norm{\vect{S}_\mathcal{F}}}{\partial \vect{r_\mu}}
    &=
    \norm{\vect{S}_\mathcal{F}} \mathbb{I}
    -
    \frac{1}{\norm{\vect{S}_\mathcal{F}}}
    \vect{S}_\mathcal{F}
    \otimes
    \vect{S}_\mathcal{F}
    \end{align}
\end{corollary}

\begin{proof}
\begin{align*}
    \sum_{\mu\in \mathcal{F}}
    \vect{r_\mu} 
    \otimes 
    \frac{\partial \norm{\vect{S}_\mathcal{F}}}{\partial \vect{r_\mu}}
    &=
    \frac{1}{\norm{\vect{S}_\mathcal{F}}}
    \sum_{\mu\in \mathcal{F}}
    \vect{r_\mu} 
    \otimes 
    \frac{\partial \vect{S}_\mathcal{F} }{\partial \vect{r_\mu}}
    \cdot
    \vect{S}_\mathcal{F}\\
    &=
    \frac{1}{\norm{\vect{S}_\mathcal{F}}}
    \left[
        \bracket{
            \vect{S}_\mathcal{F}
            \cdot
            \vect{S}_\mathcal{F}
        } \mathbb{I}
        - \vect{S}_\mathcal{F}
        \otimes
        \vect{S}_\mathcal{F}
    \right]
\end{align*}
\end{proof}

\begin{identity} 
    \begin{align} \label{identity_tensor_product_3}
        \sum_{\mu \in \mathcal{F}} 
        \vect{r_\mu}\
        \otimes
        \vect{S}_\mathcal{F} \cdot
        \frac{\partial \vect{d}_\mathcal{F}}{\partial \vect{r_\mu}}
        = \vect{d}_\mathcal{F}
        \otimes
        \vect{S}_\mathcal{F}
    \end{align}
\end{identity}

\begin{proof}
L.H.S., expressed in component form, reads:
\begin{align*}
    \sum_{\mu \in \mathcal{F}} 
    (\vect{r_\mu})^i
    (\vect{S}_\mathcal{F})^k 
    \frac{\partial (\vect{d}_\mathcal{F})^k}{\partial (\vect{r_\mu})^j}
    &= 
    \bracket{
    \vect{S}_\mathcal{F}
    }^k
    \sum_{\mu \in \mathcal{F}} 
    \bracket{
    \vect{r_\mu}
    }^i
    \frac
        {\partial
        \bracket{
            \vect{d}_\mathcal{F}
        }^k}
        {\partial
        \bracket{
        \vect{r_\mu}
        }^j}
\end{align*}
Now, since $\vect{d}_\mathcal{F}$ is an arbitrary point on the plane of face $\mathcal{F}$, we can write
\begin{align*}
    \vect{d}_\mathcal{F} = \vect{r_\nu}
    + \alpha
    \bracket{
        \vect{r_{\nu+1}}-\vect{r_{\nu}}
    }
    + \beta
    \bracket{
        \vect{r_{\nu}}-\vect{r_{\nu-1}}
    }
\end{align*}
where $\vect{r_{\nu}}$ is any of the vertices of face $\mathcal{F}$ (arranged in counter-clockwise order) and $\alpha,\beta$ are arbitrary parameters. Accordingly, 
\begin{align*}
\sum_{\mu \in \mathcal{F}} 
    \bracket{
    \vect{r_\mu}
    }^i
    \frac
        {\partial
        \bracket{
            \vect{d}_\mathcal{F}
        }^k}
        {\partial
        \bracket{
        \vect{r_\mu}
        }^j}
    &=
    \delta_{jk}
    \bracket{
        \vect{d}_\mathcal{F}
    }^i
\end{align*}
Plugging above into L.H.S. gives the required identity.
\end{proof}

\subsection{Experimental Protocols}

{\it Cell Culture.} Vimentin wild-type (Vim$+/+$) and vimentin knockout (Vim$-/-$) mouse embryonic fibroblasts (MEFs) were gifted by J. Eriksson from Abo Akademi. All cells were cultured in Dulbecco’s Modified Eagle’s Medium + 4.5g/L glucose + 2nM L-glutamine + sodium pyruvate (DMEM, Gibco) and supplemented with 10$\%$ fetal bovine serum (FBS, Hyclone), 1$\%$ non-essential amino acid (Fisher Scientific), 1$\%$ penicillin/streptomycin (Fisher Scientific) and 25 nM of HEPES (Fisher Scientific). Cell cultures were maintained at 37 degrees C and 5$\%$ $CO_2$. Cells were passaged at confluency 70-80$\%$ and medium was replenished every 3 days.\\
{\it Spheroid Culture.} The spheroid culture was created using hanging droplet methods. In brief, cells were detached from culture dish with trypsin, then centrifuged and pipetted on the lid of a 100 mm culture dish at 3000 cells/20 ul in DMEM/F12 (10$\%$ FBS, 1x Pen/Strep). Cell aggregates were incubated inverted on the lid on top of cell media at 37 degrees C with 5$\%$ $CO_2$ for 2 days before harvesting. \\
{\it Collagen Gel Preparation.} Collagen gels were prepared using a slightly modified protocol from Ibidi (cite). Collagen type I, rat tail 5 mg/ml (Ibidi) were diluted on ice to 1.5 mg/ml with 6.67$\%$ of 5x DMEM (Fisher scientific), 10$\%$ FBS in 1x DMEM. The pH was adjusted to 7.4 using 1M NaOH and sodium bicarbonate. Cell aggregates were added to the collagen solution before it solidified. The collagen solution is then pipetted on 1$\%t$ agarose coated 24-well flat bottom plates (Corning) and let to polymerize for 40 minutes in an incubator at 37 degrees C, 5$\%$ $CO_2$ and 100$\%$ humidity. DMEM complete media was added on top of the fully polymerized collagen gel. Finally, the embed spheroids in collagen gel is imaged for 48 hours (15 minutes interval) on Nikon Eclipse Ti inverted microscope using Plan Fluor 10× objective equipped with an Andor Technologies iXon em+ EMCCD camera (Andor Technologies). Sample were maintained at 37 degrees C and 5$\%$ $CO_2$ using a Tokai Hit (Tokai-Hit) stage top incubator.

\subsection{Additional maximum shear stress distributions}

The maximum shear stress distributions with their respective Gamma distribution fit parameters for $p=0.75$ and $p=0$ are presented in Figs. S1 and S2. As $p$ decreases, the maximum shear stress distribution for the solid-like spheroids becomes more broad.  In Fig. S3 we show for fluid-like spheroids a cross-section of both the maximum shear stress and the anisotropy measure at the final time of the simulation for comparison with Fig. 3. We do not observe the maximum shear stress increasing with respect to the center of the spheroid as we do in the solid-like spheroid. In other words, there appears to not be statistically-significant spatial patterning in the fluid-like case. 

\begin{figure}[htbp!]
\begin{center}
\includegraphics[width=0.3\textwidth]{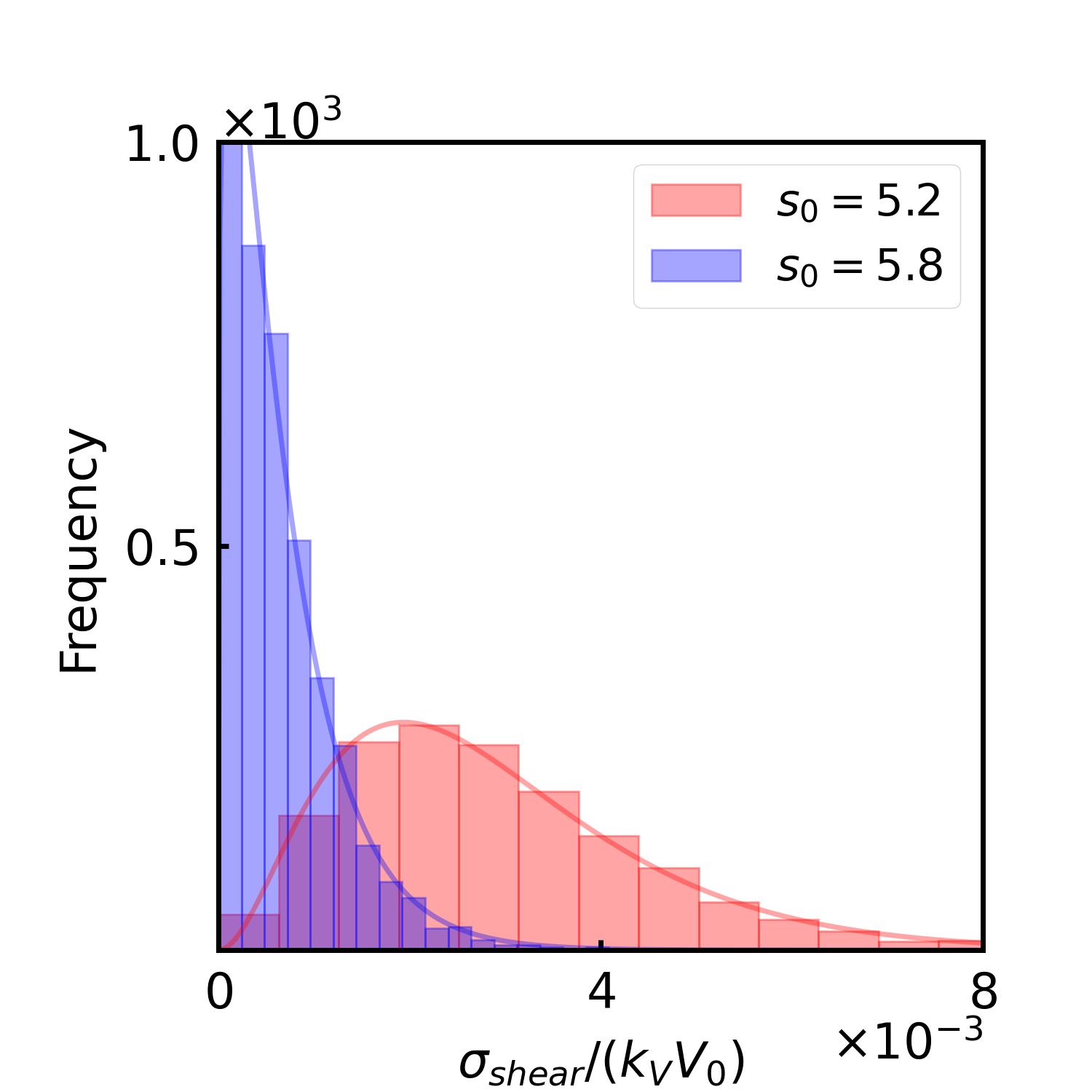}
\end{center}
\caption{{\it Histogram for maximum shear stress with $p=0.75$. Gamma fit parameters: $s_0 = 5.2$: $\alpha = 3.03, \theta = 9.53\times10^{-4}$; $s_0 = 5.8$: $\alpha = 1.24, \theta = 5.32\times10^{-4}$. }}
\label{fig:AppendixB1}
\end{figure}

\begin{figure}[htbp!]
\begin{center}
\includegraphics[width=0.3\textwidth]{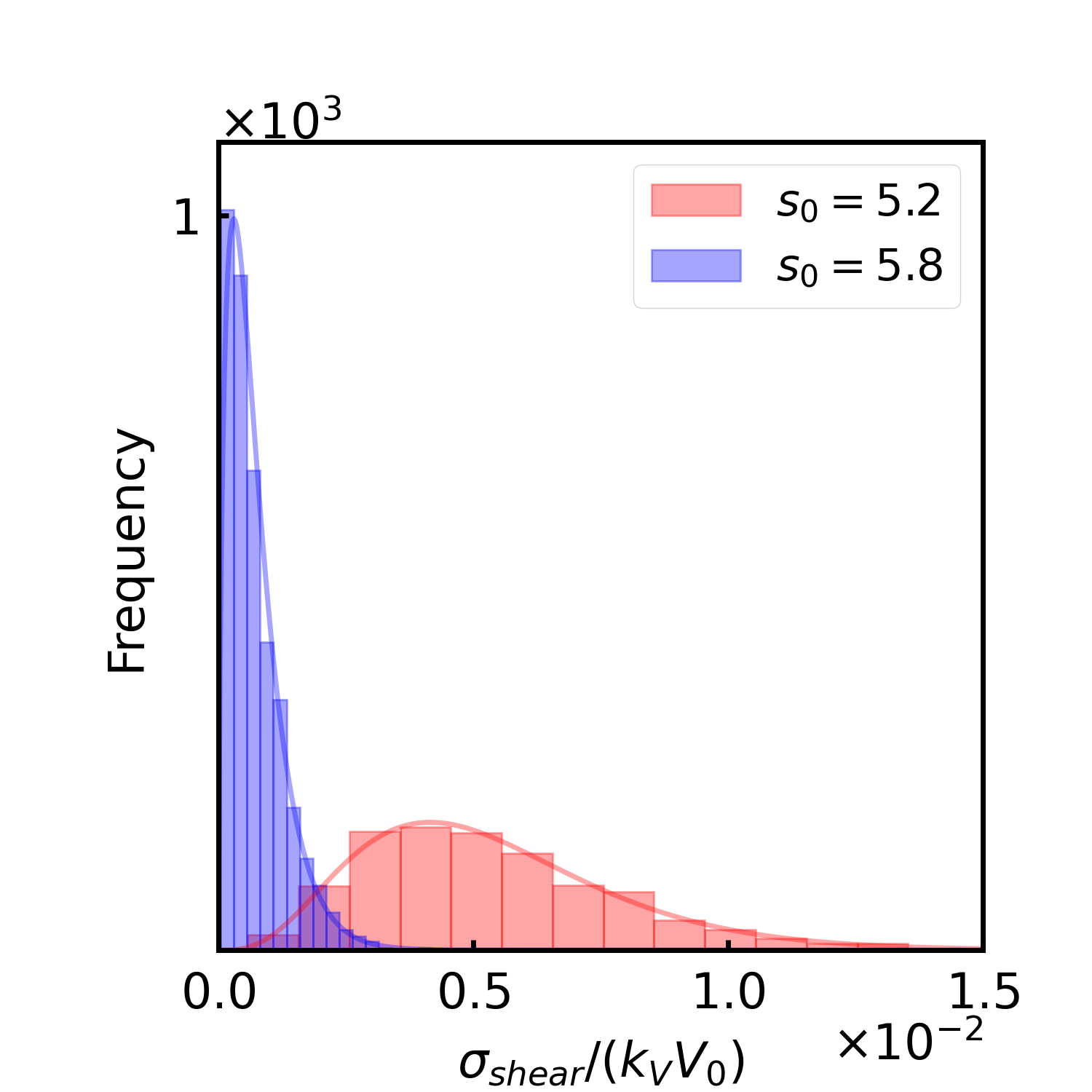}
\end{center}
\caption{{\it Histogram for maximum shear stress with no surrounding fiber network. Gamma fit parameters: $s_0 = 5.2$: $\alpha = 4.44, \theta = 1.21\times10^{-3}$; $s_0 = 5.8$: $\alpha = 1.71, \theta = 4.25\times10^{-4}$}}
\label{fig:AppendixB2}
\end{figure}

\subsection{Cell strain-stiffening and non-volume preserving deformations}
As in the main text, uni-axial strain individual cells and compute the maximum shear stress as a function of this strain. The strain is applied along the direction of the longest axis of the cell, as determined via the largest eigenvalue principle axis of the shape tensor, namely $g_3$. We apply a strain $\alpha$ on the cell via the following transformation on the positions $\vec{r_i}$ of the vertices (setting origin at the center of the cell):

\begin{align}
\vec{r_i}\to\vec{r_i}+\alpha(\hat{r_i}\cdot g_3)\hat{r_i}.
\end{align}
Note that this deformation is not volume-conserving. However, we observe similar trends as with the volume-conserving deformation with the maximum shear stress increasing nonlinearly with strain, though the results are quantitatively distinct. Please see Fig. S4. 

\begin{figure*}[t]
         \centering
         \begin{subfigure}[b]{0.4\textwidth}
             \centering
             \includegraphics[width=\textwidth]{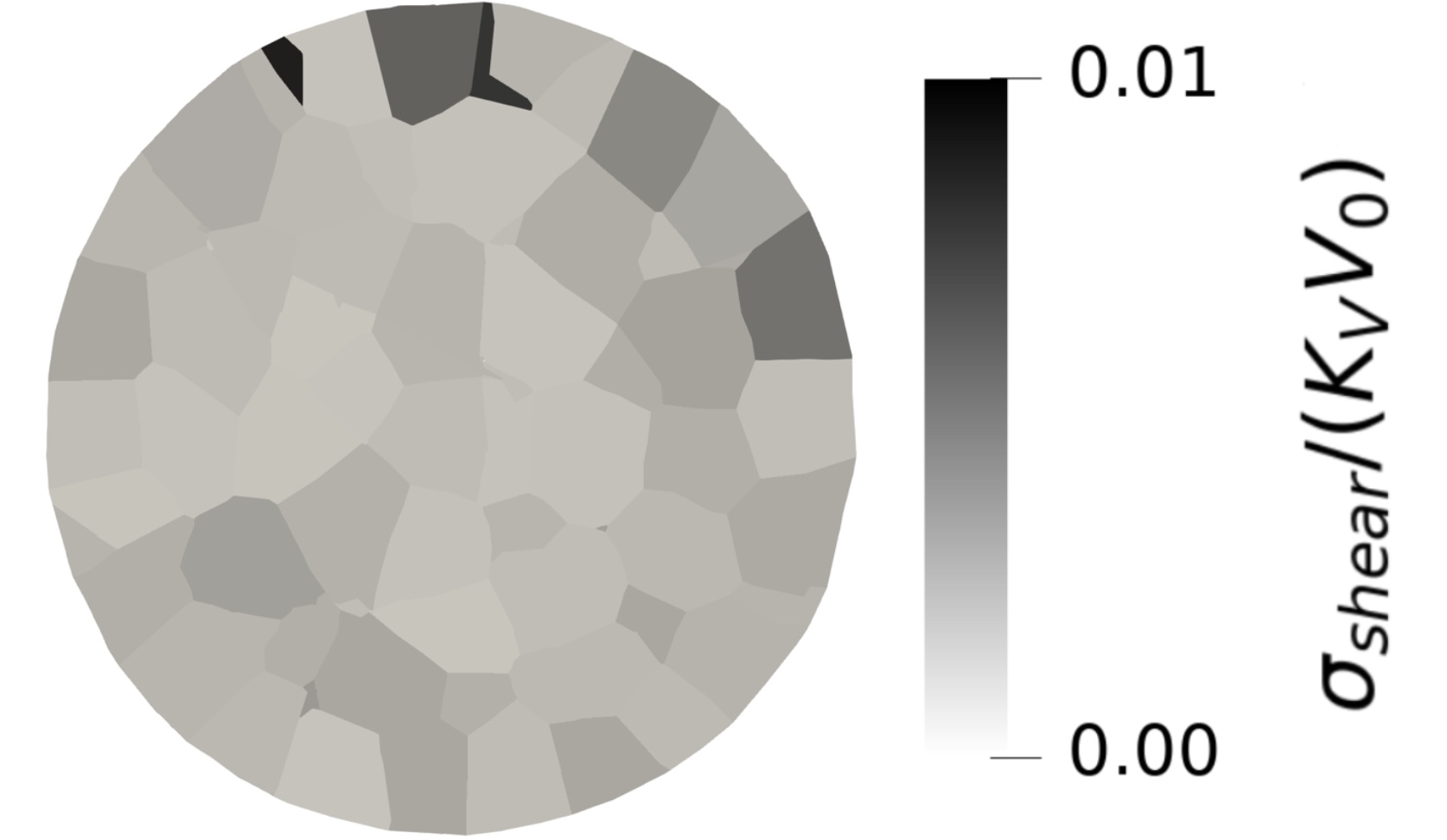}
             \caption{}
             \label{fig: shear_cross_section_fluid}
         \end{subfigure}
         \begin{subfigure}[b]{0.4\textwidth}
             \centering
             \includegraphics[width=\textwidth]{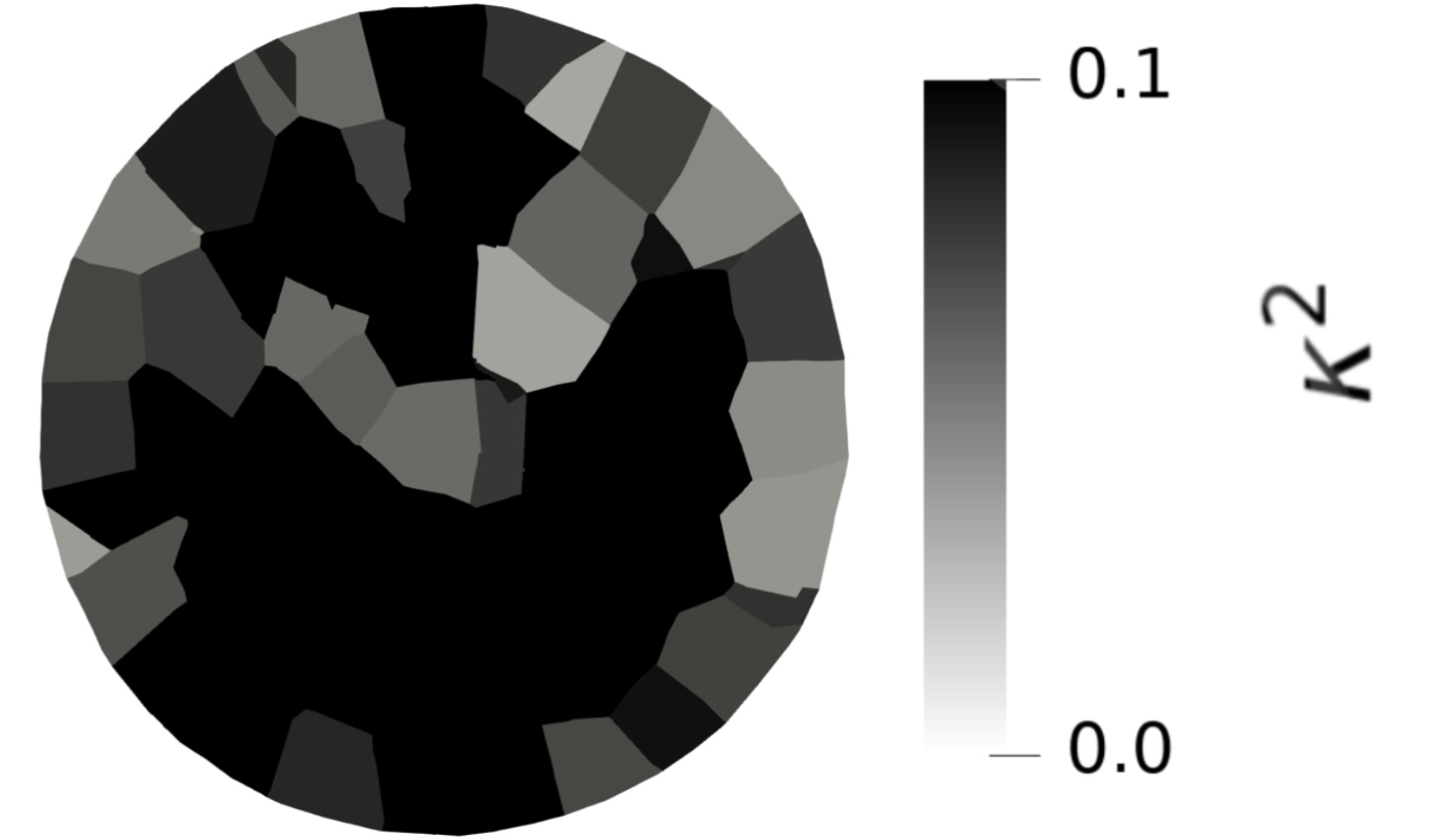}
             \caption{}
             \label{fig: anisotropy_cross_section_fluid}
         \end{subfigure}        

         \caption{{\it Spatial distribution of cellular stresses and cellular shape in fluid-like spheroids.} (a) Sample cross-section of the cellular maximum shear stress. (b) Sample cross-section of the cellular shape anisotropy.  }
\end{figure*}
\begin{figure*}[t]\label{strained_cell_comparisons}
         \centering
         \begin{subfigure}[b]{0.2\textwidth}
             \centering
             \includegraphics[width=\textwidth]{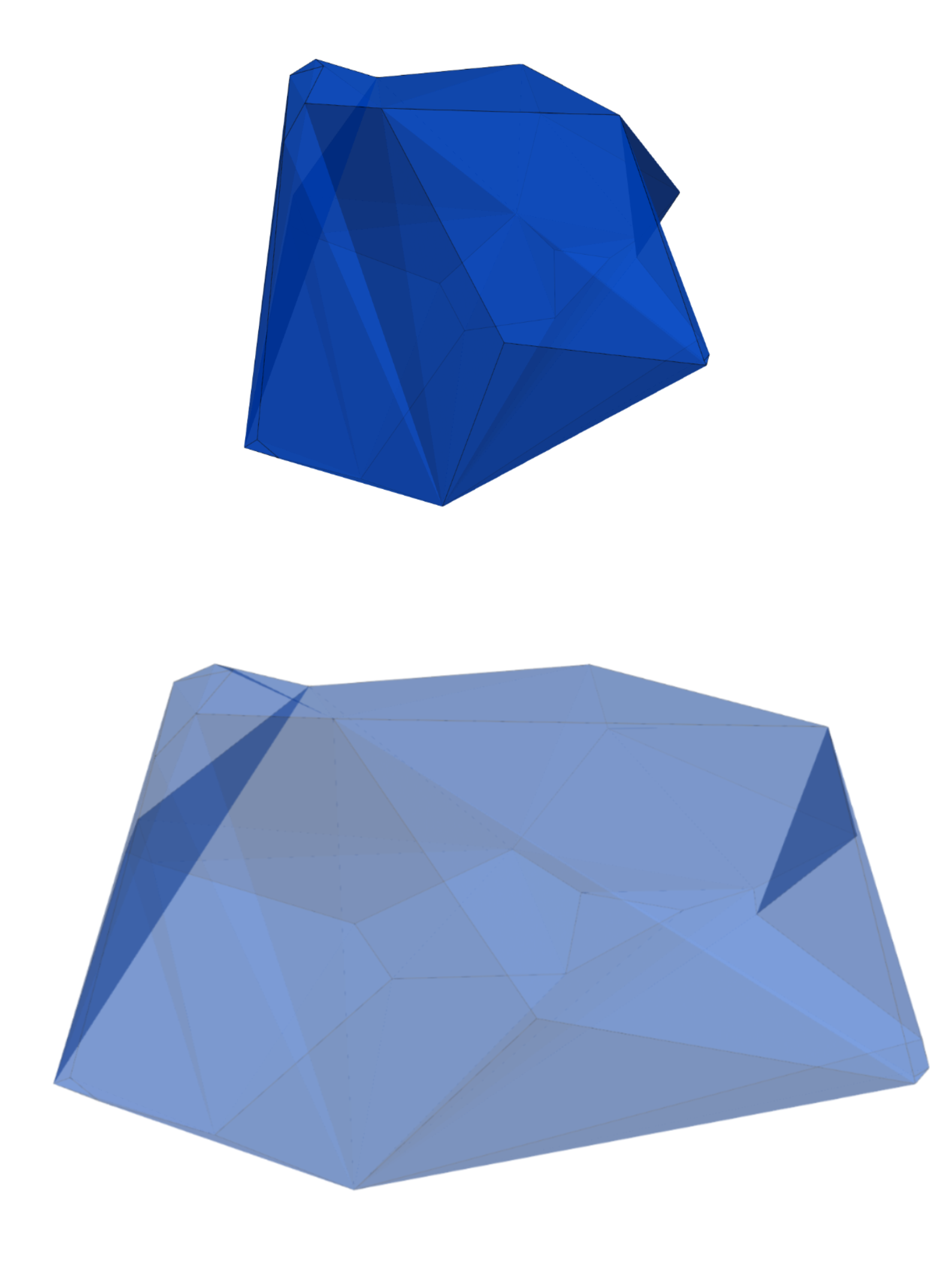}
             \caption{}
             \label{fig: final_cell}
         \end{subfigure}  
         \begin{subfigure}[b]{0.26\textwidth}
             \centering
             \includegraphics[width=\textwidth]{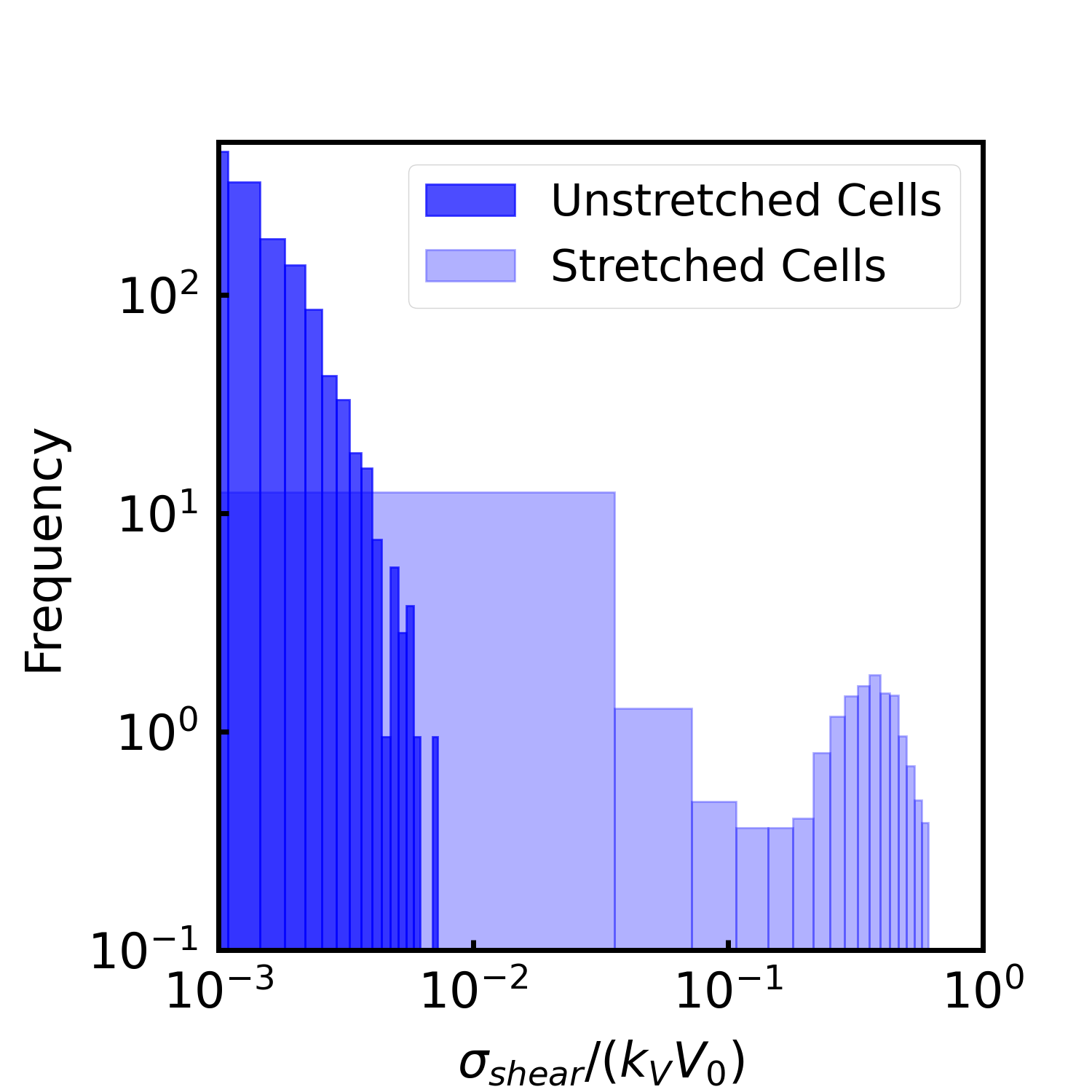}
             \caption{}
             \label{fig: single_cell_overlap_5.2}
         \end{subfigure}
         \begin{subfigure}[b]{0.26\textwidth}
             \centering
             \includegraphics[width=\textwidth]{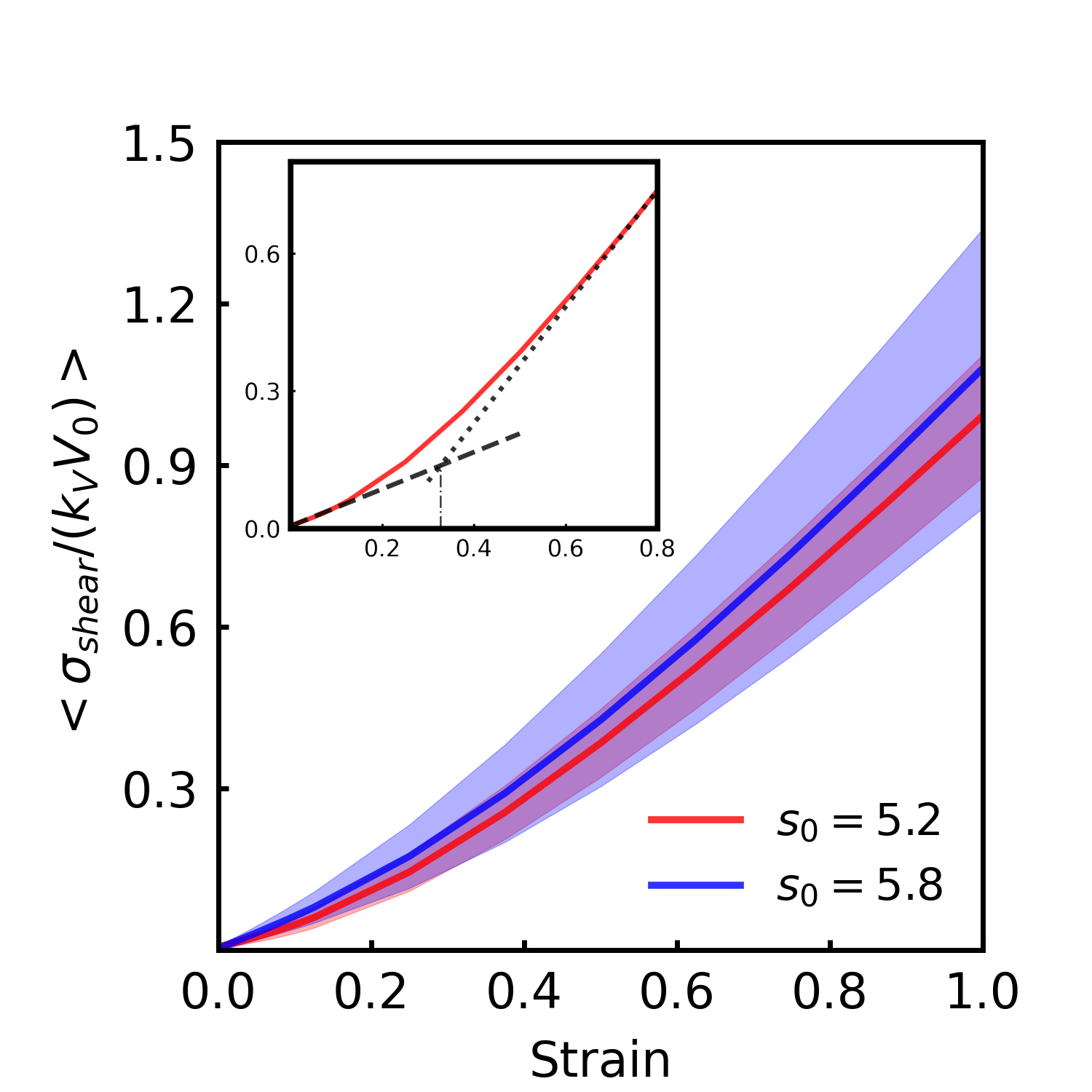}
              \caption{}
             \label{fig: single_cell_overlap_5.8}
         \end{subfigure}
         \begin{subfigure}[b]{0.26\textwidth}
             \centering
             \includegraphics[width=\textwidth]{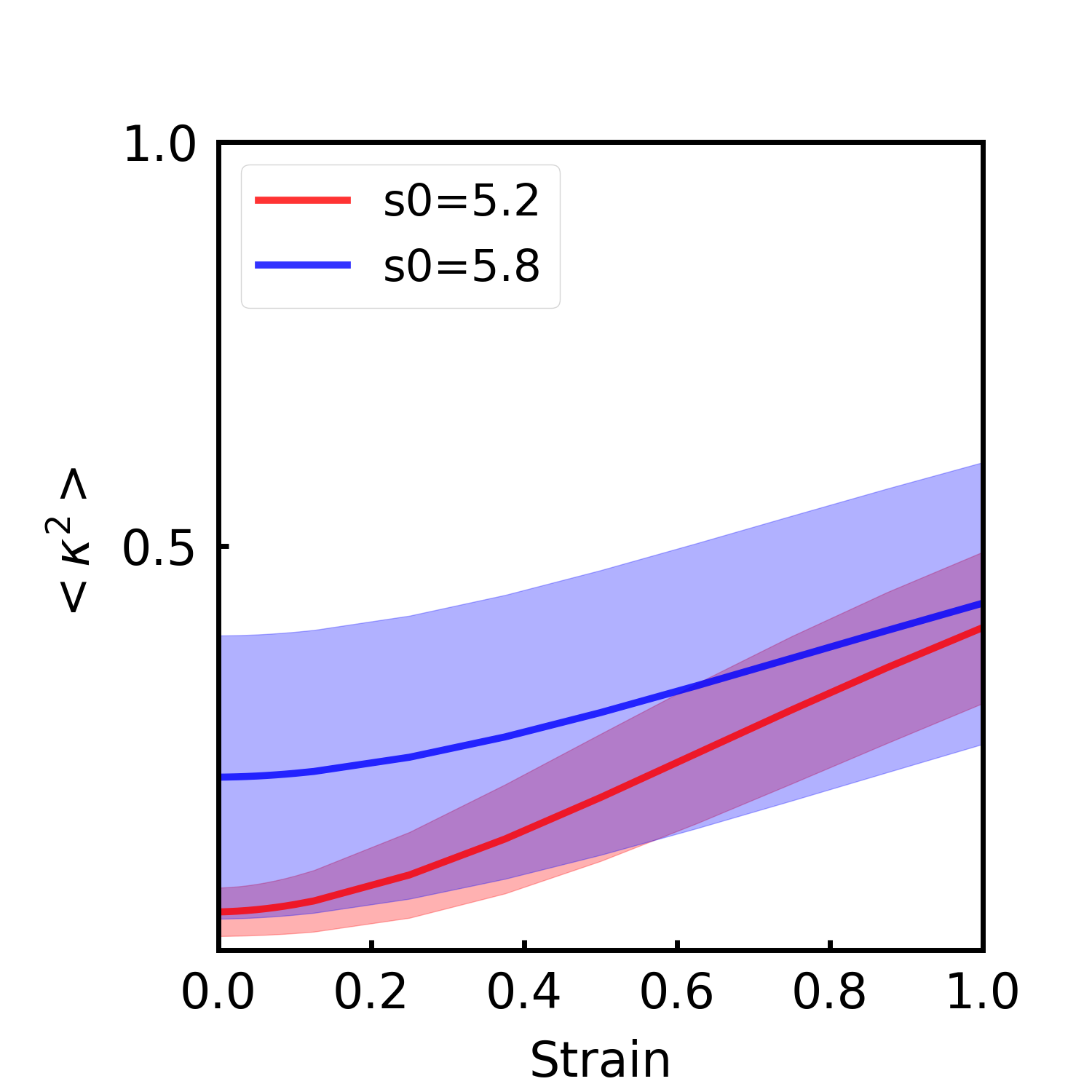}
             \caption{}
             \label{fig: stretched_cellfluid}
         \end{subfigure} 
         \caption{{\it Non-volume preserving cellular deformation}. (a) Initial cell configuration and a strained cell configuration. (b) Distribution of maximum shear stress for the initial cells (dark orange) and for the uni-axial strained cells (light orange) for $s_0=5.2$. (c) Overlap of the stress with the shape for the initial cells and for the strained cells for $s_0=5.2$. The inset is the same analysis as performed in the inset for Fig. 3(c) indicating the onset of strain stiffening at approximately a strain of 0.15. In addition, the vertical dashed line represents a crossover strain behavior from lower to higher strain. (d) Distribution of maximum shear stress for the initial cells (blue) and for the uni-axial strained cells (pink) for $s_0=5.2$. (e) Overlap of the stress with the shape for the initial cells and for the strained cells for $s_0=5.8$.  }
\end{figure*}
\subsection{Extended 3D vertex model for single- and multi-cell breakout}

We construct a three-dimensional extended vertex model, where the cells no longer share interfaces and instead interact via explicit cell-cell adhesion springs between neighboring vertices, each with spring stiffness $k_{cc}$. Stiffer cell-cell adhesion springs represent higher cell-cell adhesion and vice versa. The spheroid is initialized as a regular truncated octahedron array consisting of 12 cells with the same energy function as for the vertex model, and every vertex on the surface is connected to boundary sphere vertices to mimic fiber connections between the spheroid and the surrounding collagen, as shown in Fig.~\ref{fig:cell-breakout-results}. Here, the fibers do not contain bending energy as we are working with a more regular structure so that the bending of the fibers do not dominate the fiber mechanics. 

In this simplified model, single cell breakout is one cell breaking away from the spheroid. Moreover, cell streaming is
exhibited by an outer cell breaking away with a cell behind it moving along in the same direction. Here, the "behind" cell is the center cell. Given the complexity of cell breakout, we invoke a quasi-static approach with no active fluctuations as a starting point. We investigate how the cell or cells deform in response to being pulled out from the spheroid by decreasing the target fiber spring length associated with a particular cell at the boundary. If the distance between the cell edges increases in the single cell case, then we consider it to be breaking out. For the two-cell breakout, we analyze cell shape as well. 

Here are the details regarding the energy minimization. The simulations minimize a mechanical energy functional that describes a single
three-dimensional polyhedral cell embedded in a network of elastic interactions.
The cell is represented by a set of vertices $\{\mathbf{r}_i\}$, whose positions
are evolved quasistatically to minimize the total energy. The energy contains
three main contributions: a volume constraint, elastic penalties associated with
cell shape, and pairwise spring interactions between selected vertex pairs.

The total energy is written as
\begin{equation}
E_{\mathrm{tot}} = E_{\mathrm{vol}} + E_{\mathrm{shape}} + E_{\mathrm{pair}} .
\end{equation}

\paragraph{Volume constraint.}
To enforce approximate incompressibility of the cell, the volume $V$ of the
polyhedron is constrained to a target value $V_0$ using a quadratic penalty,
\begin{equation}
E_{\mathrm{vol}} = \frac{k_V}{2}\,(V - V_0)^2 .
\end{equation}
Here $k_V$ is a large stiffness that suppresses volume fluctuations. In the
simulations, the cell is rescaled such that $V_0 = 1$, ensuring that all
deformations are shape-changing rather than volumetric.

\paragraph{Shape elasticity.}
Cell shape is further regulated through elastic penalties associated with edge
lengths and surface geometry. For each elastic element $a$ with length $\ell_a$,
the energy contribution takes the form
\begin{equation}
E_{\mathrm{shape}} = \sum_a \frac{k_A}{2}\,\bigl(A - A_0)^2 ,
\end{equation}
where $A_0$ is the preferred cell surface area. 

\paragraph{Pairwise spring interactions.}
In addition to volume and area constraints, selected pairs of vertices interact through
pairwise springs that model either cell-cell adhesions or ECM springs. The
pairwise interaction energy is
\begin{equation}
E_{\mathrm{pair}} = \sum_{\langle i j \rangle}
\frac{k_{ij}(\ell_{ij})}{2}\,\bigl(|\mathbf{r}_i - \mathbf{r}_j| - s_0\bigr)^2 ,
\end{equation}
where $s_0$ is the rest length of the spring and $\ell_{ij} = |\mathbf{r}_i -
\mathbf{r}_j|$ is the instantaneous separation. Should the pairwise spring connect the vertices of two cells, the spring constant is denoted as $k_{cc}$, or a cell-cell adhesion spring. Should the pairwise spring connect two vertices in the fiber network, it is denoted as $k_{fn}$. Importantly, the effective spring
constant $k_{ij}$ is allowed to depend on extension. For instance, as cell-cell adhesions become unbound with greater strain, then $k_{cc}$ decreases. 

\paragraph{Parameters.} 
The dimensionless parameters implemented in the simulations are: $K_A=1$, $K_V=10$, initially $k_{cc (bo)}=0.01$ for the cell-cell adhesion springs associated with the potential breakout cell and $k_{cc}=0.5$ for the remaining cell-cell adhesion springs. For each 5 percent decrease in the equilibrium spring length for the "pulling" fiber network springs, $k_{cc(bo)}$ decreases by an order of magnitude. Not implementing a decrease in adhesion spring stiffness, given the choice of the remaining parameters, did not result in significant enhanced separation of the breakout cell. The spring constant associated with the fiber network springs is unity with the exception of the "pulling fiber networks springs with is four. We did not implement interfacial tension as we expect the straining by the fiber network springs to dominate over the interfacial tension, at least for the potential breakout cell. For anisotropic cell-cell adhesion in the breakout cell streaming case, the cell-cell adhesion springs between the outer and center cell, or $k_{cc(oc)}$, increased from 0.05 to 4.0 as the pulling fibers' equilibrium spring length is decreased by 20 percent, indicating a strengthening of cell-cell adhesion between the candidate cells.

\end{document}